%% file: main.tex
\documentclass[lettersize,journal]{IEEEtran}

\input{tex/macros}
\input{tex/symbols_boxes}

\input{tex/plotdata}

\begin{document}

\title{Privacy-Preserving Epidemiological Modeling on Mobile Graphs\iffullversion~(Full Version)$^*$\thanks{$^*$Please cite the journal version of this work published at IEEE Transactions on Information Forensics and Security (TIFS)~\cite{RIPPLE}.}\fi}

\input{tex/author_list}

\maketitle

\input{tex/0_abstract}
\section{Introduction}
\label{sec:intro}
\input{tex/1_intro}
\section{Related Work \iffullversion \& Background Information \fi}
\label{sec:prelims}
\input{tex/2_related_work}
\section{The \ourname{} Framework}
\label{sec:framework}
\input{tex/3_framework}
\section{Instantiating $\funcSimulation$}
\label{sec:instan}
\input{tex/4_instantiations}
\section{PIR-SUM: Instantiating $\funcPirSum$}
\label{sec:pirsum}
\input{tex/5_pir_sum_protocol}

\section{Evaluation}
\label{sec:evaluation}
\input{tex/6_evaluation}

\input{tex/7_acknowledgements}

\bibliographystyle{IEEEtran} 
\bibliography{references}

\appendices

\input{tex/8_app_related_work}
\input{tex/9_app_framework}
\input{tex/10_app_pir_sum_details_main.tex}

\vfill

\end{document}

%% file: tex/macros.tex
\usepackage{adjustbox} 
\usepackage{algpseudocode} 
\usepackage{amssymb,amsmath,amsthm,amsfonts}
\newtheorem{lemma}{Lemma}
\usepackage{booktabs}
\usepackage{caption}
\usepackage{cite}
\usepackage{colortbl}
\usepackage{comment}
\usepackage{enumerate}
\usepackage{enumitem}
\usepackage[para]{footmisc} 
\usepackage{graphicx}
\PassOptionsToPackage{hyphens}{url}\usepackage{hyperref}
\usepackage[hyphens]{url}
\hypersetup{
	colorlinks   = true,
	citecolor    = magenta
}
\usepackage{makecell}
\usepackage[framemethod=tikz]{mdframed} 
\usepackage{multirow}
\usepackage{nicematrix} 
\usepackage{ragged2e}
\usepackage{stfloats} 
\usepackage{tabularx}
\usepackage{wrapfig}
\usepackage{xcolor}
\usepackage{wasysym}
\usepackage{numprint}
\usepackage{subcaption}
\usepackage{pgfplots}
\pgfplotsset{compat=1.18}

\def\myurl#1{\setbox0\vbox{\hsize.5\maxdimen
\url{#1}\par
\setbox0\lastbox
\global\setbox1\hbox{\unhbox0\unskip\unskip\unpenalty}}\unhbox1 }

\NiceMatrixOptions
{
    custom-line =
    {letter = I ,
    tikz = {line width=0.5pt},
    width=8pt   
    },
}

\newif\ifeprint
\eprintfalse 

\newif\ifsubmission
 \submissiontrue 

\newif\iffullversion
\fullversionfalse
\fullversiontrue

\newif\ifanonymous
\anonymousfalse 

\ifsubmission
	\newcommand{\TODO}[1]{}
	
\else
	\newcommand{\TODO}[1]{\marginpar{\textcolor{red}{TODO: #1}}}
		
\fi

\let\oldS\S
\newcommand{\sect}[1]{\oldS\ref{#1}}
\newcommand{\fig}[1]{Fig.~\ref{#1}}

\newcommand{\tab}[1]{Tab.~\ref{#1}}

\newcommand{\ci}[1]{~\cite{#1}}

\newcommand{\circledfig}[1]{\tikz[baseline=(char.base)]{\node[shape=circle,draw,inner sep=0.5pt] (char) {#1};}}

\newcommand{\bfcircled}[1]{\circledfig{\textbf{#1}}}

\newcommand{\accept}{\texttt{Accept}}

\newcommand{\abort}{\texttt{abort}}

\newcommand{\ourname}{{RIPPLE}} 
\newcommand{\teename}{{RIPPLE\textsubscript{\sf TEE}}}
\newcommand{\pirname}{{RIPPLE\textsubscript{\sf PIR}}}

\newcommand{\participantset}{\mathcal{P}}
\newcommand{\participant}[1]{\mathcal{P}_{#1}}
\newcommand{\numparticipants}{\mathsf{p}}
\newcommand{\numservers}{\mathsf{M}}
\newcommand{\mpcservers}{\ensuremath{\mathcal{C}}}
\newcommand{\mpcserver}[1]{\ensuremath{\mathsf{S}_{#1}}}
\newcommand{\ri}{\ensuremath{\mathsf{RI}}}

\newcommand{\funcCollection}{\ensuremath{\mathcal{F}_{{\sf gen}}}}
\newcommand{\funcSimulation}{\ensuremath{\mathcal{F}_{{\sf esim}}}}
\newcommand{\funcAggregation}{\ensuremath{\mathcal{F}_{{\sf agg}}}}
\newcommand{\funcAnon}{\ensuremath{\mathcal{F}_{{\sf anon}}}}
\newcommand{\funcPirSum}{\ensuremath{\mathcal{F}_{{\sf pirsum}}}}

\newcommand{\funcPirS}{\ensuremath{\mathcal{F}^{\sf 2S}_{\sf pir}}}
\newcommand{\funcVerify}{\ensuremath{\mathcal{F}_{{\sf vrfy}}}}

\newcommand{\PIRSUM}{\ensuremath{\mathsf{PIR_{sum}}}}
\newcommand{\PIRSUMA}{\ensuremath{\mathsf{PIR^{I}_{sum}}}}
\newcommand{\PIRSUMB}{\ensuremath{\mathsf{PIR^{II}_{sum}}}}

\newcommand{\numberofsimulations}{\ensuremath{N_{\sf sim}}}
\newcommand{\simsteps}{\ensuremath{N_{\sf step}}}
\newcommand{\currstep}{\ensuremath{\mathsf{s}}}
\newcommand{\simcurrstep}{\ensuremath{\mathsf{s}_{\sf sim}}}

\newcommand{\infclass}{\mathsf{class_{inf}}}
\newcommand{\numinfclass}{\ensuremath{N_{\sf inf}}}
\newcommand{\infclassval}[1]{\mathsf{class^{#1}_{inf}}}

\newcommand{\simclassval}[2]{\mathsf{v^{#1}_{#2}}}
\newcommand{\aggclassval}[1]{\mathsf{C^{#1}_{inf}}}

\newcommand{\likelihood}{\delta}
\newcommand{\rlikelihood}{\hat{\delta}} 
\newcommand{\clikelihood}{\Delta}
\newcommand{\blindedclikelihood}{\hat{\Delta}}

\newcommand{\encounterset}[1]{\ensuremath{\mathcal{E}_{#1}}}
\newcommand{\maxencounters}[1]{\ensuremath{{E}^{\sf max}_{#1}}}
\newcommand{\avgencounters}[1]{\ensuremath{{E}^{\sf avg}_{#1}}}
\newcommand{\hybrid}[1]{\ensuremath{\textbf{HYB}_{\sf #1}}}

\newcommand{\entrynode}{\ensuremath{\mathcal{N}_{\texttt{entry}}}}
\newcommand{\exitnode}{\ensuremath{\mathcal{N}_{\texttt{exit}}}}

\newcommand{\pkey}{\ensuremath{\mathrm{pk}}}
\newcommand{\skey}{\ensuremath{\mathrm{sk}}}

\newcommand{\Enc}[1]{\ensuremath{\mathsf{Enc}_{#1}}}

\newcommand{\query}{q}
\newcommand{\querySet}{\mathcal{Q}}

\newcommand{\qshift}{q'}
\newcommand{\qshiftcorr}[1]{\ensuremath{\theta}_{#1}}

\newcommand{\threshold}{\tau}

\newcommand{\db}{\ensuremath{\mathsf{D}}} 
\newcommand{\dbV}[1]{\ensuremath{\mathsf{D}[#1]}} 
\newcommand{\dbsize}{N}

\newcommand{\maskeddb}[1]{\ensuremath{\mathsf{D}^{#1}}} 
\newcommand{\maskeddbV}[2]{\ensuremath{\maskeddb{#1}[#2]}}

\newcommand{\pirresult}{\ensuremath{\mathsf{res}}} 
\newcommand{\pirmask}{m}
\newcommand{\pirmaskquery}{\ensuremath{\mathsf{r}}} 
\newcommand{\pirkey}{k}

\newcommand{\Gen}{{\sf Gen}}
\newcommand{\Eval}{{\sf Eval}}
\newcommand{\Ver}{{\sf Ver}}

\newcommand{\Adv}{\mathcal{A}}

\newcommand{\ashr}[2]{\ensuremath{\langle #1 \rangle}_{#2}}

\newcommand{\bitvector}[1]{\ensuremath{\vec{\mathsf{b}^{#1}}}}
\newcommand{\bitV}[2]{\ensuremath{\mathsf{b}^{#1}_{#2}}}

\newcommand{\Z}[1]{\ensuremath{\mathbb{Z}}_{2^{#1}}}

\newcommand{\Hash}{\ensuremath{\mathsf{H}}}
\newcommand{\permutation}{\ensuremath{\pi}}

\definecolor{darkred}{rgb}{.6,0,0}
\definecolor{darkgreen}{rgb}{0,.4,0}
\definecolor{darkblue}{rgb}{0,0,.6}

\newcommand{\xor}{\oplus}
\newcommand{\Xor}{\bigoplus}

\newcommand{\makeparafit}{\looseness=-1}

\setlist[description]{style=unboxed,leftmargin=0cm}

\newenvironment{tifs_enumerate}{
	\begin{list}{{$\bullet$}}{
			\setlength\partopsep{0pt}
			\setlength\parskip{0pt}
			\setlength\parsep{2pt}
			\setlength\topsep{0pt}
			\setlength\itemsep{1pt}
			\setlength{\itemindent}{5pt}
			\setlength{\leftmargin}{1pt}
		}
	}{
	\end{list}
}

\newenvironment{tifs_innerenumerate}{
	\begin{list}{{$\bullet$}}{
			\setlength{\itemindent}{15pt}
			\setlength{\leftmargin}{0pt}
		}
	}{
	\end{list}
}

\newenvironment{tifs_item_text}{
	\begin{list}{{$\bullet$}}{
			\setlength\partopsep{2pt}
			\setlength\parskip{0pt}
			\setlength\parsep{2pt}
			\setlength\topsep{3pt}
			\setlength\itemsep{1pt}
			\setlength{\itemindent}{20pt}
			\setlength{\leftmargin}{0pt}
		}
	}{
	\end{list}
}

\newenvironment{tifs_item_prot}{
	\begin{list}{{$\bullet$}}{
			\setlength\partopsep{2pt}
			\setlength\parskip{0pt}
			\setlength\parsep{2pt}
			\setlength\topsep{3pt}
			\setlength\itemsep{0pt}
			\setlength{\itemindent}{5pt}
			\setlength{\leftmargin}{3pt}
		}
	}{
	\end{list}
}

\newenvironment{tifs_innerprot}{
	\begin{list}{{$\bullet$}}{
			\setlength\partopsep{2pt}
			\setlength\parskip{0pt}
			\setlength\parsep{1pt}
			\setlength\topsep{2pt}
			\setlength\itemsep{1pt}
			\setlength{\itemindent}{12pt}
			\setlength{\leftmargin}{10pt}
		}
	}{
	\end{list}
}

\newenvironment{acm_item_text}{
	\begin{list}{{$\bullet$}}{
			\setlength\partopsep{2pt}
			\setlength\parskip{0pt}
			\setlength\parsep{2pt}
			\setlength\topsep{3pt}
			\setlength\itemsep{1pt}
			\setlength{\itemindent}{20pt}
			\setlength{\leftmargin}{0pt}
		}
	}{
	\end{list}
}

\newenvironment{acm_inneritem}{
	\begin{list}{{$\bullet$}}{
			\setlength\partopsep{2pt}
			\setlength\parskip{0pt}
			\setlength\parsep{1pt}
			\setlength\topsep{2pt}
			\setlength\itemsep{1pt}
			\setlength{\itemindent}{16pt}
			\setlength{\leftmargin}{8pt}
		}
	}{
	\end{list}
}

\newenvironment{acm_text_list}{
	\begin{list}{{$\bullet$}}{
			\setlength\partopsep{2pt}
			\setlength\parskip{0pt}
			\setlength\parsep{2pt}
			\setlength\topsep{3pt}
			\setlength\itemsep{1pt}
			\setlength{\itemindent}{5pt}
			\setlength{\leftmargin}{3pt}
		}
	}{
	\end{list}
}

\newenvironment{acm_text_sublist}{
	\begin{list}{{$\bullet$}}{
			\setlength\partopsep{2pt}
			\setlength\parskip{0pt}
			\setlength\parsep{2pt}
			\setlength\topsep{3pt}
			\setlength\itemsep{1pt}
			\setlength{\itemindent}{12pt}
			\setlength{\leftmargin}{10pt}
		}
	}{
	\end{list}
}

\definecolor{LightGray}{gray}{0.97}

\newcommand{\signature}{\sigma}
\newcommand{\securityparam}{\lambda}

\newcommand{\metadata}{m}
\newcommand{\encounterlist}{E}
\newcommand{\rencounterlist}{R}

\newcommand{\infectionclass}{I}
\newcommand{\initialinfectionclass}{I^{\mathsf{init}}}

\newcommand{\token}{r}

\newcommand{\user}{\mathcal{P}}

\newcommand{\psim}{\mathsf{param_{sim}}}

\newcommand{\covid}{COVID-19}

\newtheorem{definition}{Definition}

%% file: tex/author_list.tex
\author{
    \IEEEauthorblockN{
        Daniel Günther\IEEEauthorrefmark{2}, 
        Marco Holz\IEEEauthorrefmark{2}, 
        Benjamin Judkewitz\IEEEauthorrefmark{3}, 
        Hellen Möllering\IEEEauthorrefmark{2}, 
        Benny Pinkas\IEEEauthorrefmark{4}, 
        Thomas Schneider\IEEEauthorrefmark{2}, 
        Ajith Suresh\IEEEauthorrefmark{5}
    }
    
    \IEEEauthorblockA{
        \IEEEauthorrefmark{2}Technical University of Darmstadt, Germany
        \IEEEauthorrefmark{3}Charité-Universitätsmedizin, Germany
        \IEEEauthorrefmark{4}Bar-Ilan University, Israel
        \IEEEauthorrefmark{5}Technology Innovation Institute~(TII), Abu Dhabi
    }
    \iffullversion\IEEEauthorblockA{\\\{\href{mailto:guenther@encrypto.cs.tu-darmstadt.de}{guenther}, 
    \href{mailto:holz@encrypto.cs.tu-darmstadt.de}{holz},
    \href{mailto:moellering@encrypto.cs.tu-darmstadt.de}{moellering}, 
    \href{mailto:guenther@encrypto.cs.tu-darmstadt.de}{schneider}\}@encrypto.cs.tu-darmstadt.de,
    \href{mailto:benjamin.judkewitz@charite.de}{benjamin.judkewitz@charite.de},
    \href{mailto:benny@pinkas.net}{benny@pinkas.net},
    \href{mailto:ajith.suresh@tii.ae}{ajith.suresh@tii.ae}}\fi  
}

%% file: tex/0_abstract.tex
\begin{abstract}
\makeparafit
The latest pandemic \covid{} brought governments worldwide to use various containment measures to control its spread, such as contact tracing, social distance regulations, and curfews. Epidemiological simulations are commonly used to assess the impact of those policies before they are implemented. Unfortunately, the scarcity of relevant empirical data, specifically detailed social contact graphs, hampered their predictive accuracy. As this data is inherently privacy-critical, a method is urgently needed to perform powerful epidemiological simulations on real-world contact graphs without disclosing any sensitive~information.

\makeparafit
In this work, we present RIPPLE, a privacy-preserving epidemiological modeling framework enabling standard models for infectious disease on a population's real contact graph while keeping all contact information locally on the participants' devices. As a building block of independent interest, we present PIR-SUM, a novel extension to private information retrieval for secure download of element sums from a database. Our protocols are supported by a proof-of-concept implementation, demonstrating a 2-week simulation over half a million participants completed in 7 minutes, with each participant communicating less than 50 KB.
\end{abstract}

\begin{IEEEkeywords}
	epidemiological modeling, private information retrieval, secure multi-party computation
\end{IEEEkeywords}

%% file: tex/1_intro.tex
The \covid{} pandemic has profoundly impacted daily life, leading to heightened mental illness and domestic abuse cases~\iffullversion\cite{vindegaard2020covid,maison2021challenges,taub2020new}\else\cite{maison2021challenges}\fi.
Governments globally have taken steps, such as lockdowns and institution closures, to curb the virus while supporting the economy. Despite these measures, infections surged, and many lives were lost. Afterwards, new infectious diseases like monkeypox have spread, resulting in quarantines in Europe~\iffullversion\cite{caulcutt,cooney2022,ecdpc2022}\else\cite{cooney2022}\fi.

In the context of \covid, contact tracing apps were used all over the world to notify contacts of potential infections~\cite{reichert2020privacy,ahmed2020survey,TroncosoPHSLLSP20,vaudenay2020,trieu2020epione, pinkas2021hashomer,LueksGVBSPT21,hogan2021contact}. 
Unfortunately, there is a fundamental limitation to contact tracing: It only notifies contacts of an infected person \emph{after} the infection has been detected, i.e., typically after a person develops symptoms, is tested, receives the test result, and can connect with contacts~\cite{tupper2021,lewis2020}.
Tupper et al.~\cite{tupper2021} report that in British Columbia in April 2021, this process ideally took five days, reducing new cases by only 8\% compared to not using contact tracing. They conclude that contact tracing must be supplemented with multiple additional containment measures to control disease spread effectively. 

To overcome these inherent limitations of contract tracing, we consider epidemiological modeling, which allows predicting the spread of an infectious disease in the \emph{future}, as done in~\cite{giordano2020modelling,kucharski2020effectiveness,cooper2020sir,silva2020covid,shinde2020forecasting,ThompsonEpid20,SEIRD21}. 
It allows to assess the effectiveness of containment measures by mathematically modeling their impact on the spread. As a result, it can be an extremely valuable tool for governments to select effective containment measures~\cite{ThompsonEpid20}. For example, Davis et al.~\cite{davies2020effects} predicted in early 2020 that \covid{} would infect $85$\% of the British population without any containment measures in place, causing a massive overload of the health system ($13$-$80\times$ the capacity of intensive care units). 
Their forecast also indicated that short-term interventions, such as school closures and social distancing, would not effectively reduce the number of cases. As a result, the British government decided to implement a lockdown in March 2020, effectively reducing transmissions and stabilizing the health system~\cite{ThompsonEpid20}.

Access to detailed information about a population's size, density, transportation, and health care system enables accurate epidemiological modeling to forecast disease transmission in various scenarios~\cite{adam2020special}. Precise, up-to-date data on movements and physical interactions is crucial for forecasting transmission and assessing the impact of control measures before implementation~\cite{klepac2020contacts}. In practice, these simulations can quickly model the spread of a disease, project the number of infections based on specific actions, and predict how the disease might spread to specific areas.

However, data on personal encounters is scarce, limiting accurate assessments of containment measures' impacts~\iffullversion\cite{adam2020special,ferguson05,klepac2020contacts}\else\cite{adam2020special}\fi. This scarcity arises because so far encounter data is often obtained through surveys, which fail to capture the reality of random encounters in public places~\iffullversion\cite{klepac2020contacts,edmunds1997mixes}\else\cite{klepac2020contacts}\fi. Additionally, social interaction patterns can change rapidly, as seen with social distancing measures, making collected data quickly outdated. Therefore, existing data cannot realistically simulate person-to-person social contact graphs. Ideally, epidemiologists need a complete physical interaction graph of the population, but strict privacy regulations make accurate tracking of interpersonal contacts unacceptable.

To address the issue of preserving privacy while obtaining up-to-date contact data, we present \ourname, a practical framework for epidemiological modeling that allows precise disease spread simulations using current contact information, incorporating control measures without leaking information about individuals' contacts. \ourname{} provides a privacy-preserving method for collecting real-time physical encounters and can compute arbitrary compartment-based epidemiological models\footnote{The implementation of concrete simulation functions is outside the scope of this work and referred to medical experts. More details on epidemiological modeling are given in~\sect{sec:prelims}.} on the latest contact graph of the previous days. \ourname{} is not only applicable to \covid, but to \emph{any} infectious diseases. We expect that \ourname's strong privacy guarantees will encourage wider participation, enabling epidemiologists to conduct more accurate simulations and develop effective containment strategies. 

Additionally, given the success and public acceptance of contact-tracing apps during \covid, users are familiar with the value of digital tools in reducing disease spread~\cite{JMIR:BHHCGS23,EJPH:CovidReview}. \ourname{} builds on this trust by offering a privacy-preserving solution that can be seamlessly integrated into existing contact-tracing apps, thus expanding beyond individual contact tracing to support community-wide epidemiological modeling. Similar to traditional contact-tracing apps, \ourname{} involves only a minimal exchange of two tokens when two participants encounter each other. To minimize any impact on device performance, we propose that the final simulation of disease spread is executed at night while most participants sleep and their phones are charging. This timing ensures that \ourname{} runs unobtrusively on participants’ mobile devices and preserves a smooth user experience. Our work mainly focuses on providing an initial cryptographic solution for privacy-preserving epidemiological modeling and leaves challenges for practical deployment as future work.

\medskip
\subsubsection*{Our Contributions}
This paper introduces \ourname~(cf. \figref{overview-framework}), a framework for expanding the scope of privacy research from contact tracing to epidemiological modeling. While the former only warns about potential infections in the past, epidemiological modeling can predict the spread of infectious diseases in the future. Anticipating the effects of various control measures allows for the development of informed epidemic containment strategies and political interventions before their implementation.

\begin{figure}[htb!]
	\centering
	\includegraphics[width=1\columnwidth]{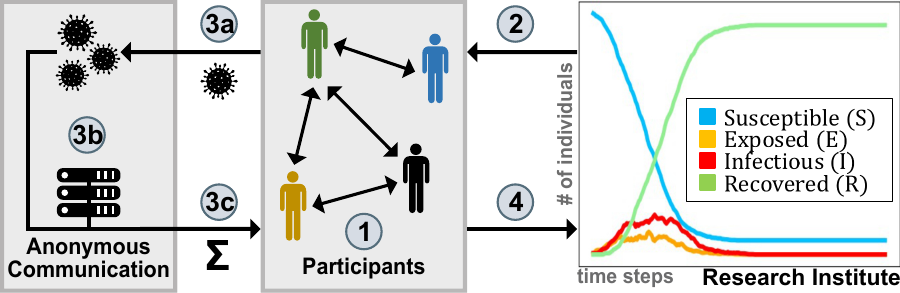}
	{\small 
        \renewcommand*{\arraystretch}{0.1}
		\begin{tabular}{p{0.96\linewidth}}
            \\[3pt]
			\rowcolor{LightGray}
			\bfcircled{1}~Mobile apps collect anonymous encounter tokens during interactions.
			\bfcircled{2}~Research Institute begins the simulation by providing initialization parameters.
			\bfcircled{3{\footnotesize a}}~Participants securely upload infection likelihood to servers.
			\bfcircled{3{\footnotesize b}}~Servers securely compute cumulative infection likelihood per participant.
			\bfcircled{3{\footnotesize c}}~Participants retrieve their cumulative infection likelihood.
			\bfcircled{4}~The aggregate results (\#S,\#E,\#I,\#R) are sent to the Research Institute.\\
	\end{tabular}}
	\caption{Overview of \ourname{} Framework.}%
	\label{fig:overview-framework}
\end{figure}

\makeparafit
\ourname{} uses a fully decentralised system similar to the federated learning paradigm~\cite{mcmahan2017} to achieve high acceptance and trust in the system and to motivate many participants to join the system to generate representative contact information. All participant data, such as encounter location, time, and distance, are kept locally on the participants' devices. Participants in \ourname{} communicate through anonymous communication channels enabled by a group of \emph{semi-honest} central servers.

\ourname{} is instantiated with two methods for achieving privacy-preserving epidemiological modeling, each covering a different use case. The first is \teename, which assumes each participant's mobile device has a Trusted Execution Environment~(TEE).  
The second method is \pirname, which eliminates this assumption by utilizing cryptographic primitives such as Private Information Retrieval~(PIR). Along the way, we develop a multi-server PIR extension that allows a client to retrieve the sum of a set of elements (in our case, infection likelihoods) from a database without learning individual entries.

We assess the practicality of our methods by benchmarking core building blocks using a proof of concept implementation. Our findings indicate that, with adequate hardware, both protocols can scale up to millions of participants. For instance, a simulation of 14 days with 1 million participants can be completed in less than half an hour.

Our contributions are summarized as follows:
\begin{enumerate}[wide]
    \item We present \ourname, the \emph{first} privacy-preserving framework for epidemiological modeling on contact information stored on mobile devices.
    \item \ourname{} formalizes the notion of \emph{privacy-preserving} epidemiological modeling and defines privacy requirements in the presence of both semi-honest and malicious participants.
    \iffullversion
    \item We present two techniques -- \teename{} and \pirname{} -- that combine anonymous communication techniques with either TEEs or PIR and anonymous credentials.
    \else
    \item We present \pirname{}, an instantiation of \ourname{} that combines anonymous communication techniques with PIR and anonymous credentials.\footnote{A TEE-based technique is presented in the full version~\cite{ARXIV:RIPPLE}.} 
    \fi
    \item We propose PIR-SUM, an extension to existing PIR schemes, that allows a client to download the sum of $\threshold$ distinct database entries without learning the values of individual entries or revealing which entries were requested.
    \item We demonstrate the practicality of our framework by providing an open source implementation and a detailed performance evaluation of \ourname{}.
\end{enumerate}

%% file: tex/2_related_work.tex
\iffullversion

This section discusses related works addressing privacy challenges in the context of infectious diseases as well as necessary background information on contact tracing and epidemiological modeling (including a clarification of the differences between the two). An overview of the (cryptographic) primitives and other techniques used in this work is presented in \sect{app:related-primitives}.

\subsection{Cryptography-based Solutions in the Context of Infectious Diseases}
CrowdNotifier~\cite{LueksGVBSPT21} notifies visitors of (large) events about an infection risk when another visitor reported SARS-CoV-2 positive after the event, even if they have not been in close proximity of less than 2 meters. To protect user privacy, it follows a distributed approach where location and time information is stored encrypted on the user's device. Bampoulidis et al.~\cite{BHKRW20} introduce a privacy-preserving two-party set intersection protocol that detects infection hotspots by intersecting infected patients, input by a health institute, with customer data from mobile network operators. 
 
CoVault~\cite{DeViti2022} is a data analytics platform based on secure multi-party computation techniques (MPC) and trusted execution environments. The authors discuss the usage of CoVault for storing location and timing information of people usable by epidemiologists to analyse (unique) encounter frequencies or linkages among two disease outbreak clusters while preserving privacy.

Al-Turjman and David Deebak~\cite{Al-TurjmanD20} integrate privacy-protecting health monitoring into a Medical Things device that monitors the health status~(heart rate, oxygen saturation, temperature, etc.) of users in quarantine with moderate symptoms. Only in the event of an emergency is medical personnel notified. Pezzutto et al.~\cite{PezzuttoRSG21} optimize the distribution of a limited set of tests to identify as many positive cases as possible, which are then isolated. Their system can be deployed in a decentralized, privacy-preserving environment to identify individuals who are at high risk of infection. Barocchi et al.~\cite{BarsocchiCCDFGM21} develop a privacy-preserving architecture for indoor social distancing based on a privacy-preserving access control system. When users visit public facilities (e.g., a supermarket or an airport), their mobile devices display a route recommendation for the building that maximizes the distance to other people. Bozdemir et al.~\cite{bozdemir2021privacy} suggest privacy-preserving trajectory clustering to identify typical movements of people and detect forbidden gatherings when contact restrictions are in place.

\paragraph{Contact Tracing}\label{sec:contacttracing} A plethora of contact tracing systems has been introduced and deployed since the outbreak of the pandemic~\cite{reichert2020privacy,ahmed2020survey,ciucci2020national}. They either use people's location (GPS or telecommunication provider information) or measure proximity (via Bluetooth LE). Most systems can be categorized into centralized and decentralized designs~\cite{vaudenay2020}. In a centralized contact tracing system (e.g., \cite{singapur,robert}), computations such as the generation of the tokens exchanged during physical encounters are done by a central party. This central party may also store some contact information depending on the concrete system design. In contrast, in decentralized approaches (e.g., \cite{TroncosoPHSLLSP20,pinkas2021hashomer,chan2020}), computation and encounter information remain (almost completely) locally at the participants' devices. 

Contact tracing focuses on determining contacts of infected people in the past. In contrast, epidemiological modeling, which we consider in this work, forecasts the spread of infectious diseases in the future. Thus, epidemiological modeling goes \emph{beyond} established contact tracing systems. They share some technical similarities (specifically, the exchange of encounter tokens), but on top of anomalously recording the contact graph, simulations have to be run on it. Similarly, presence tracing and hotspot detection are concerned with ``flattening the curve'' in relation to infections in the past. In contrast, epidemiological modeling is a tool for decision-makers to evaluate the efficacy of containment measures like social distancing in the future, allowing them to ``get ahead of the wave''.

\subsection{Epidemiological Modeling}
\else
This section discusses related works on disease modeling and privacy-preserving epidemiological modeling.
\fi
\label{sec:epimodel} 

\paragraph{Disease Modeling} \makeparafit  There are various ways to model a disease mathematically\iffullversion\cite{brauer2008compartmental,brauer2019simple,harko2014exact,Schlickeiser2021,fernandez2022estimating,gray2011stochastic,SEIRD21,he2020seir}\else\cite{brauer2008compartmental,he2020seir,TBD:GMPS21,IDM:AAY21}\fi. 
Popular compartment models use a few continuous variables linked by differential equations to capture disease spread. In the SEIR model~\ci{kermack1991,he2020seir}, individuals are assigned to four compartments: susceptible (S), exposed (E), infectious (I), and recovered (R). These models are useful for understanding macroscopic trends and are widely used in epidemiological research~\cite{small2005,chen2020}. However, they condense complex individual behaviors into a few variables, limiting predictive power~\cite{may2021}.
In contrast, agent-based epidemiological models~\cite{ferguson2006strategies} simulate the spread by initializing numerous agents with individual properties (e.g., location, age) and interaction rules. This allows for more realistic disease transmission simulations by modeling individual behaviors. Combining agent-based models with compartment models enhances the realism and accuracy of disease forecasting. Simulations with varying parameters, like interaction reductions or targeted vaccinations, are run to predict the effects of different policy interventions.

\makeparafit
A key challenge is modeling agents' contact behaviors. Older models used survey-based contact matrices to estimate average contacts within age ranges~\cite{klepac2020contacts}, which improved over uniform assumptions but still fell short. Aggregated network statistics can't replicate real network dynamics, including super-spreaders with numerous contacts~\cite{Kupferschmidt2020}. Ideally, epidemiologists would like to use real-world contact graphs of all individuals, but this is often challenging due to privacy concerns.

\paragraph{Contact Tracing for Privacy-Preserving Epidemiological Modeling} 
If contact information collected through contact tracing apps was centralized, an up-to-date full contact graph could be constructed for epidemiological simulations~\cite{TBD:GMPS21,IDM:AAY21}. However, contact information is highly sensitive and should not be shared. Contact information collected via mobile phones can reveal who, when, and whom people meet, which is sensitive and must be protected. Beyond, such information also enables to derive indications about the financial situation~\cite{luo2017inferring} and personality~\cite{montjoye2013predicting}. One can think about many more examples: By knowing which medical experts are visited by a person, information about the health condition can be anticipated; contact with members of a religious minority as well as visits to places related to religion might reveal a religious orientation, etc. Thus, it would be ideal for enabling precise epidemiological simulations without leaking individual contact information.

\makeparafit
One way to achieve privacy-preserving epidemiological modeling using contact tracing apps is by allowing each participant's device to share its contact information secretly with a set of non-colluding servers. These servers can then run simulations using secure multi-party computation (MPC). Araki et al.~\cite{Araki0OPRT21} demonstrated efficiently running graph algorithms on secret shared graphs via MPC. However, despite the common non-collusion assumption in the crypto community, public trust issues may arise if all contact information is disclosed once servers collude. To address this, \ourname{} distributes trust by enabling participants to keep their contact information local while anonymously sending messages to each other to simulate the disease spread. Only aggregated simulation results are shared with research institutes, ensuring no direct identity or contact data is disclosed. This method resembles Federated Learning~\cite{mcmahan2017} and the contact tracing designs by Apple and Google.\footnote{https://covid19.apple.com/contacttracing} This distributed design can increase trust and facilitate broad adoption of the system.

To the best of our knowledge, \ourname{} is the first framework that allows executing any agent-based compartment model on the distributed real contact graph while maintaining privacy.

%% file: tex/3_framework.tex
\makeparafit
\ourname's primary goal is to enable the evaluation of the impact of multiple combinations of potential containment measures defined by epidemiologists and the government, and to find a balance between the drawbacks and benefits of those measures, rather than to deploy the measures in ``real-life" first and then analyse the impact afterwards. Such measures may include, for example, the requirement to wear face masks in public places, restrictions on the number of people allowed to congregate, the closure of specific institutions and stores, or even complete curfews and lockdowns within specific regions.

\makeparafit
Participants in \ourname{} collect personal encounter data anonymously and store it locally on their mobile devices, similar to privacy-preserving contact tracing apps. For epidemiological modeling, \ourname{} must derive a contact graph without leaking sensitive personal information to simulate disease spread over a specific period, like two weeks. Most countries have a 6-hour period at night when most people are asleep, and their mobile devices are idle, connected to WiFi, and possibly charging—an ideal time for running \ourname{} simulations. Medical experts can then analyze the results to understand the disease better, and political decision-makers can identify the most effective containment measures to implement.

To acquire representative and up-to-date physical encounter data, widespread public usage of \ourname{} would be ideal. One way to encourage this is to piggyback \ourname{} on most countries' official contact tracing applications. On the other hand, politicians can motivate residents beyond the intrinsic incentive of supporting public health by coupling the use of \ourname{} with additional benefits such as discounted or free travel passes.

\subsection{System and Threat Model}\label{sec:threatmodel}
\ourname{} comprises of $\numparticipants$ participants, denoted collectively by $\participantset$, a research institute $\ri$ who is in charge of the epidemiological simulations, and a set of MPC servers $\mpcservers{}$ responsible for anonymous communication among the participants. 
\tab{tab:parameters} summarizes the notations used in this work.

\input{tex/notations_table}

We assume that the research institute and MPC servers are semi-honest~\cite{goldreich2009foundations}, meaning they follow protocol specifications correctly while attempting to gather additional information. These semi-honest MPC servers also establish an anonymous communication channel. 
\iffullversion
We discuss the security of the anonymous communication channel in more detail in~\sect{app:anon-comm-channel}.
\fi
A protocol is secure if nothing is leaked beyond what can be inferred from the output. While the semi-honest security model is not the strongest, it offers a good trade-off between privacy and efficiency, making it popular in practical privacy-preserving applications such as privacy-preserving machine learning~\cite{ASTRA,mishra2020delphi,ABY2}\iffullversion, genome/medical research~\cite{veeningen2018enabling,tkachenko2018large,schneider2019episode}, and localization services~\cite{jarvinen2019pilot,BeetsNS22}.\else ~and genome/medical research~\cite{tkachenko2018large,schneider2019episode}\fi. This model protects against passive attacks by curious administrators and accidental data leakage and often serves as a foundation for developing protocols with stronger privacy guarantees~\cite{lindell2007efficient,aumann2010security}. We consider this a reasonable assumption, as the research institute and servers will be controlled by generally trusted entities such as governments or public medical research centers, potentially collaborating with NGOs like the EFF\footnote{\url{https://www.eff.org}} or the CCC\footnote{\url{https://www.ccc.de/en/}}.

Given the widespread interest in discovering effective containment measures, we expect a high level of intrinsic motivation among participants for successful epidemiological modeling. 
However, we cannot assume that millions of potential participants are behaving honestly. Therefore, our set of participants~$\participantset$ can include clients who are malicious and may attempt to compromise the privacy of the scheme. This includes scenarios where some clients deviate from the protocol to gain unauthorized information. This stronger security notion, referred to as \emph{malicious-privacy}, has been the focus of several recent works~\cite{lehmkuhl2021,chandran2021} and provides a higher level of privacy than semi-honest corruptions.

Our primary focus is to protect against malicious clients who want to gain unauthorized information. A stronger threat involves malicious clients undermining the correctness of the protocol, thereby disrupting or compromising the accuracy of the simulation. Addressing the issue of malicious clients targeting correctness is left as future work. Such a correct behaviour could for example be enforced using trusted execution environments like ARM Trust Zone available in many smartphones.

\subsection{Phases of \ourname}
\label{sec:phases-framework}
\ourname{} is divided into four phases as shown in \fig{fig:overview-framework}: i) Token Generation, ii) Simulation Initialization, iii) Simulation Execution, and iv) Result Aggregation. While our framework can be applied to any compartment-based epidemiological modeling of any infectious disease~(cf.~\sect{sec:epimodel}), we explain \ourname{} using the prevalent Covid-19 virus and the SEIR model~\cite{kermack1991,diekmann2012} as a running example. For simplicity, we assume that an app that emulates \ourname{} is installed on each participant's mobile device and that the participants locally enter attributes such as workplace, school, regular eateries, and cafes in the app after installing the app.

\boxref{fig:sim-framework} summarises the phases of the \ourname{} framework in the context of a single simulation setting and we give details below. Multiple simulations can be executed in parallel. The concrete number of simulation runs with the same parameters or different parameters should be determined by epidemiologists. Note that simulations are run on collected data, e.g., from the last days, and not on real-time encounter information. This combines efficiency requirements with maximally up-to-date encounter information.

\iffullversion\input{fig/boxes/ripple}\fi

\begin{tifs_enumerate}
   \item[{\bf \bfcircled{1} - Token Generation:}] During a physical encounter, participants exchange data via Bluetooth LE to collect anonymous encounter information\iffullversion ~(\fig{fig:token-collection})\fi, similar to contact tracing~\cite{pinkas2021hashomer,hatke2020using}. These tokens are stored locally on the users' devices and do not reveal any sensitive information (i.e., identifying information) about the individuals involved. In addition to these tokens, the underlying application will collect additional information on the context of the encounter, known as ``metadata" for simulation purposes. This varies depending on the underlying instantiation of the protocol and can include details such as duration, proximity, time, and location. The metadata can include or exclude different encounters in the simulation phase, allowing the effect of containment measures to be modeled (e.g., restaurant closings by excluding all encounters that happened in restaurants). The token generation phase is not dependent on the simulation phase, so no simulation-dependent infection data is exchanged. The token generation phase is modelled as an ideal functionality $\funcCollection$ that will be instantiated later in \sect{sec:instan}.

\iffullversion
    \begin{figure}[htb!]
        \centering
        \begin{subfigure}{0.45\textwidth}
          \centering
          \includegraphics[width=0.5\linewidth]{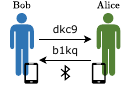}
          \caption{Token Generation}
          \label{fig:token-collection}
        \end{subfigure}
        \begin{subfigure}{0.45\textwidth}
          \centering
          {{\includegraphics[width=0.6\linewidth]{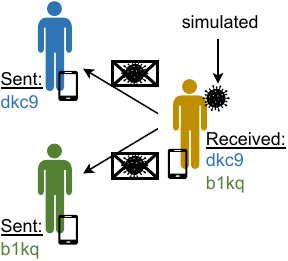}}}
          \caption{Simulation}
          \label{fig:simulation-execution}
        \end{subfigure}
        \caption{Token Generation and Simulation phases in \ourname{}.}\label{fig:test}
    \end{figure}

    \fi

\iffullversion
   \emph{Running Example:} Assume that a participant, Alice, takes the bus to pick up her daughter from school. There are several other people on this bus -- for simplicity, we call them Bob\textsubscript{1}$,\ldots,$ Bob\textsubscript{x}. As part of the token generation phase, Alice's phone exchanges unique anonymous tokens with the devices of the different Bobs. Now, two weeks later, it is night, and the national research institute ($\ri$) wants to run a simulation covering 14 days to see how closing all schools would affect the spread of the disease. To accomplish this, the $\ri$ notifies all registered participants' applications to run a simulation using encounter data from the previous two weeks.
   \fi

    \iffullversion\else\input{fig/boxes/ripple}\fi
   
   \item[{\bf \bfcircled{2} - Simulation Initialization:}] The research institute $\ri$ initiates the simulation phase by sending a set of parameters, denoted by $\psim$, to the participants in $\participantset$. The goal is to ``spread" a fictitious infection across $\numberofsimulations$ different simulation settings. To begin a simulation, each participant $\participant{i}$ is assigned to an infection class $\initialinfectionclass_i \in \infclass$ (e.g., \{S\}usceptible, \{E\}xposed, \{I\}nfectious, \{R\}ecovered for the SEIR model) as specified in $\psim$. For each individual simulation, $\psim$ defines a set of containment measures, such as school closings and work from home, which the participants will use as filters to carry out the simulation in the next stage. In addition, $\ri$ publishes a formula to calculate the infection likelihood $\likelihood$. The likelihood is determined by several parameters in the underlying modeling, such as encounter distance and time. For example, this likelihood might range from 0~(no chance of infection) to 100~(certain to get infected).

\iffullversion
   \emph{Running Example:} Assume Alice is designated as infectious, while Bob\textsubscript{1} is designated as susceptible by $\ri$. The other participants Bob\textsubscript{2}$,\ldots,$ Bob\textsubscript{x} are also assigned to an infection class (S, E, I, or R). To simulate containment measures, the \ourname-app now employs filters defined in $\psim$. Using the information provided by the participants\footnote{This may also include location data obtained from the mobile app., e.g., Check In and Journal fields in the Corona-Warn contact tracing app.}, the application may automatically filter out encounters that would not happen if a containment measure were in place, such as encounters in school while simulating school closings. 
   \fi
   \item[{\bf \bfcircled{3} - Simulation Execution:}] Once $\ri$ initialises the simulation, $\simsteps$~simulation steps~(steps~\bfcircled{3{\footnotesize a}}, \bfcircled{3{\footnotesize b}}, \bfcircled{3{\footnotesize c}} in~\fig{fig:overview-framework}) are performed for each of the~$\numberofsimulations$ simulation settings~(e.g.,~$\simsteps=14$~days). Without loss of generality, consider the first simulation step and let $\numberofsimulations = 1$. The simulation proceeds as follows:
    \begin{tifs_innerenumerate}
    	\item[1)] Participant $\participant{i} \in \participantset$ filters out the relevant encounters based on the containment measures defined by $\ri$. Let the set $\encounterset{i}$ represent the corresponding encounter tokens.
    	\item[2)] For each token $\token_e \in \encounterset{i}$, $\participant{i}$ computes the infection likelihood~$\likelihood_i^{\token_e}$ using the formula from $\ri$, i.e., the probability that~$\participant{i}$ infects the participant met during the encounter with identifier token~$\token_e$. 
    	\item[3)] Participants use the likelihood values $\likelihood$ obtained in the previous step to execute an ideal functionality called $\funcSimulation$, which allows them to communicate the $\likelihood$ values anonymously through a set of MPC servers $\mpcservers$. Furthermore, it allows each participant $\participant{j}$ to receive a cumulative infection likelihood, denoted by $\clikelihood_j$, based on all of the encounters they had on the day being simulated, i.e., $\clikelihood_j = \sum_{\token_e \in \encounterset{j}} \rlikelihood_j^{\token_e}$.  In this case, $\rlikelihood_j^{\token_e}$ denotes the infection likelihood computed by participant $\participant{f}$ and communicated to $\participant{j}$ for an encounter between $\participant{f}$ and $\participant{j}$ with identifier token~${\token_e}$. As will be discussed later in \sect{sec:privacy-framework}, $\funcSimulation$ must output the cumulative result rather than individual infection likelihoods because the latter can result in a breach of privacy. 
    	\item[4)] \makeparafit Following the guidelines set by the $\ri$, $\participant{j}$ updates its infection class $\infectionclass_j$ using the cumulative infection likelihood $\clikelihood_j$ acquired in the previous step.
        \end{tifs_innerenumerate}
    These steps above are repeated for each of the $\simsteps$ simulation steps in order and across all the $\numberofsimulations$ simulation settings.

\iffullversion
    \emph{Running Example:} Let the simulated containment measure be the closure of schools. As Alice is simulated to be infectious, Alice's phone computes the infection likelihood for every single encounter it recorded on the day exactly two weeks ago (Day 1) \emph{except} those that occurred at her daughter's school. Then, Alice's phone combines the computed likelihood of each encounter with the corresponding unique encounter token to form tuples, which are then sent to the servers instantiating the anonymous communication channel. Using the encounter token as an address, the servers anonymously forward the likelihood to the person Alice has met, for example, Bob\textsubscript{1} (cf.~\fig{fig:simulation-execution}). Likewise, Bob\textsubscript{1} receives one message from each of the other participants he encountered and obtains the corresponding likelihood information. Bob\textsubscript{1} aggregates all likelihoods he obtained from his encounters on Day~1 and checks the aggregated result to a threshold defined by the $\ri$ to see if he has been infected in the simulation\footnote{Bob\textsubscript{1} obtains the aggregated likelihood in the actual protocol.}. 
    \fi
   \item[{\bf \bfcircled{4} - Result Aggregation:}] For a given simulation setting, each participant $\participant{i} \in \participantset$ will have its infection class $\infectionclass_i^{\currstep}$ updated at the end of every simulation step ${\currstep} \in [\simsteps]$. This phase allows $\ri$ to obtain each simulated time step's aggregated number of participants per class (e.g., \#S, \#E, \#I, \#R). For this, we rely on a \emph{Secure Aggregation} functionality, denoted by $\funcAggregation$, which takes a $\numinfclass$-tuple of the form $\{\simclassval{1}{i},\ldots,\simclassval{\numinfclass}{i}\}^{\currstep}$ from each participant for every simulation step~${\currstep}$ and outputs the aggregate of this tuple over all the $\numparticipants$ participants to $\ri$. In this case, $\simclassval{k}{i}$ is an indicator variable for the $k$-th infection class, which is set to one if $\infectionclass_i^{\currstep} = \infclassval{k}$ and zero otherwise. Secure aggregation~\cite{erkin2013,kursawe2011,li2010,li2010} is a common problem in cryptography these days, particularly in the context of federated learning, and there are numerous solutions proposed for various settings, such as using TEEs, a semi-trusted server aggregating ciphertexts under homomorphic encryption, or multiple non-colluding servers that aggregate secret shares. In this work, we consider $\funcAggregation$ a black box that can be instantiated using existing solutions compatible with our framework.

\iffullversion
   \emph{Running Example:} All participants will know their updated infection class at the end of Day 1's simulation round, and they will prepare a 4-tuple of the form $\{\simclassval{S}{}, \simclassval{E}{}, \simclassval{I}{}, \simclassval{R}{} \}$ representing their updated infection class in the SEIR model. Participants will then engage in a secure aggregation protocol that determines the number of participants assigned to each infection class, which is then delivered to the $\ri$. Then, the second simulation round begins, replicating the procedure but using encounters from 13 days ago, i.e., Day~2. The $\ri$ holds the aggregated number of participants per day per class after simulating all 14 days, i.e., a simulation of how the disease would spread if all schools had been closed in the previous 14 days (cf. graph in~\fig{fig:overview-framework}).
   \fi
\end{tifs_enumerate}

\subsection{Privacy Requirements}
\label{sec:privacy-framework}
\makeparafit
A private contact graph necessitates that participants remain unaware of any unconscious interactions. This means they cannot determine if they had unconscious contact with the same person more than once or how often they did, which is considered a privacy issue, exemplified in the next paragraph. We remark that an insecure variant of \ourname, in which each participant $\participant{i}$ receives the infection likelihood $\rlikelihood_i^e$ for all its encounters $e \in \encounterlist_i$ separately, will not meet this condition.

\paragraph{\bf Linking Identities Attacks} To demonstrate this, observe that when running multiple simulations (with different simulation parameters $\psim$) on the same day, participants will use the same encounter tokens and metadata from the token generation phase in each simulation. If a participant $\participant{i}$ (Alice) can see the infection likelihood $\rlikelihood_i$ of each of her interactions, $\participant{i}$ can look for correlations between those likelihoods to see if another participant $\participant{j}$ (Bob) was encountered more than once. We call this a \emph{Linking Identities Attack}~and depict it in \fig{fig:linking-identities}, where, for simplicity, the infection likelihood accepts just two values: 1 for high and 0 for low infection likelihood.

\setlength{\columnsep}{8pt}%
\setlength{\intextsep}{8pt}%
\begin{wrapfigure}{r}{0.26\textwidth}
	\centering
	\includegraphics[clip, trim=0pt 0pt 0pt 0pt, width=0.25\textwidth]{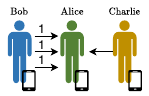}
	\captionsetup{font={small}}
	\caption{Linking Identities Attack. Alice and Bob had several encounters, but Alice and Charlie only had one.}%
	\label{fig:linking-identities}
\end{wrapfigure}

Consider the following scenario to help clarify the issue: Alice and Bob work together in the same office. As a result, they have numerous conscious encounters during working hours. However, in their spare time, they may be unaware that they are in the same location (e.g., a club) and may not want the other to know. Their phones constantly collect encounters even if they do not see each other. Assume the research institute~($\ri$) sent the participants a simple infection likelihood formula that returns 0 (not infected) or 1 (infected). Furthermore, since the data is symmetric, Alice and Bob have the same metadata (duration, distance, etc.) about their conscious and unconscious encounters. 
Let Bob be modelled as infectious in the first simulation. As a result, he will send a 1 for each (conscious and unconscious) encounter he had (including those with Alice). If multiple simulations are run on the same day (i.e., with the same encounters), Alice will notice that some encounters, specifically all conscious and unconscious encounters with Bob, always have the same infection likelihood: If Bob is not infectious, all will return a 0; if Bob is infectious, all will return a 1. Thus, even if Alice had unconscious encounters with Bob, she can detect the correlations between the encounters and, as a result, determine which unconscious encounters were most likely with Bob.

\makeparafit
The more simulations she runs, the more confident she becomes. Since every participant knows the formula, this attack can also be extended to complex infection likelihood functions. While it may be more computationally expensive than the simple case, Alice can still identify correlations. This attack works even if all of the encounters are unconscious. In such situations, Alice may be unable to trace related encounters to a single person (Bob), but she can infer that they were all with the same person (which is more than learning nothing).

\makeparafit
To avoid a Linking Identities attack, \ourname{} ensures that each participant receives an aggregation of all infection likelihoods of their encounters during the simulation. It cannot be avoided that participants understand that when ``getting infected'' some of their contacts must have been in contact with a (simulated) infectious participant. As this is only a simulated infection, we consider this leakage acceptable.

In summary, the Linking Identities attack demonstrates that repeated simulations with identical metadata (tokens) enable participants to infer connections between conscious and unconscious encounters. In \ourname, participants never learn individual infection likelihoods, but only an aggregated infection likelihood, which highly decreases the practicality of this attack.

\iffullversion\else
\setlength{\columnsep}{8pt}%
\setlength{\intextsep}{8pt}%
\begin{wrapfigure}{r}{0.28\textwidth}
	\centering
	\includegraphics[clip, trim=0pt 0pt 0pt 0pt, width=0.24\textwidth]{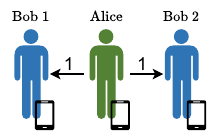}
	\captionsetup{font={small}}
	\caption{Sybil Attack.}%
	\label{fig:sybil-attack}
\end{wrapfigure}
\fi
\paragraph{\bf Sybil Attack} While the Linking Identities Attack is already significant in the semi-honest security model, malicious participants may further circumvent aggregation mechanisms that prevent access to individual infection likelihoods. They could, for example, construct many \emph{sybils}, i.e., multiple identities using several mobile devices, to collect each encounter one by one and then conduct a Linking Identities Attack with the information.

\iffullversion
\setlength{\columnsep}{8pt}%
\setlength{\intextsep}{8pt}%
\begin{wrapfigure}{r}{0.28\textwidth}
	\centering
	\includegraphics[clip, trim=0pt 0pt 0pt 0pt, width=0.24\textwidth]{fig/figure6.pdf}
	\captionsetup{font={small}}
	\caption{Sybil Attack.}%
	\label{fig:sybil-attack}
\end{wrapfigure}
\fi

Apart from the privacy threat, the sybil attack can affect the accuracy of the simulation as potentially more devices than people in one location will participate in the simulation, which does not correspond to reality. If a malicious participant who exploits a sybil attack is simulated as infectious, their visited locations will tend to spread the simulated infection more significantly and thus classify the risk of infection as higher than usual.

A registration system can be used to increase the costs of performing sybil attacks, i.e., to prevent an adversary from creating many identities. This assures that only legitimate users can join and participate in the simulation. In a closed ecosystem, such as a company, this can be achieved by letting each member receive exactly one token to participate in the simulation. On a larger scale at the national level, one can let each citizen receive a token linked to a digital ID card. In such authentication mechanisms, anonymous credentials \iffullversion (cf.~\sect{app:related-primitives})\fi can be used to ensure anonymity, and we leave the problem for future work, as our initial work on privacy-preserving epidemiological modeling focuses on malicious clients targeting privacy, but not yet correctness~(cf.~\sect{sec:threatmodel}).

To summarize, from a privacy perspective, Sybil attacks enable a Linking Identities Attack by allowing malicious participants to create multiple identities. Furthermore, they have an impact on the accuracy of the simulation.

\paragraph{\bf Inference Attacks} 
\makeparafit
Note that although \ourname{} mimics the spirit of Federated Learning (FL)~\cite{mcmahan2017}, it is not susceptible to so-called inference attacks~\cite{nasr2019comprehensive,fereidooni2021safelearn} in the same sense as FL. First, \ourname{} only reveals the final output (to a research institute $\ri$) and no individual updates/results that ease information extraction. We, however, note that the analysis results provided to $\ri$ (cf.~\sect{sec:threatmodel}) contain information about the spread of the modeled disease in a specific population (otherwise, it would be meaningless to run the simulation). The ideal functionality does not cover leakage from the final output but protects privacy during the computation. Thus, our security model does not consider anything that might be inferred from the output. We also argue that it is in the public interest to provide such aggregated information to the $\ri$ for deciding upon effective containment measures against infectious diseases.

In summary, while \ourname{} mitigates inference attacks by restricting individual data leakage, the final aggregated output may reveal trends relevant to public health, balancing privacy with the societal need for actionable insights.

%% file: tex/notations_table.tex
\begin{table}[htb!]
	\resizebox{\columnwidth}{!}{
        \centering
		\begin{NiceTabular}{ r c l }
			\toprule
			& ~Parameter & Description \\
			\midrule
			\Block[r]{3-1}{\rotatebox{90}{Entities}} 
			& $\participantset$   & Set of all participants; $\participantset = \{\participant{1}, \ldots, \participant{\numparticipants}\}$\\
			&$\ri$                         & Research Institute \\ 
			& $\mpcservers{}$     & Communication Servers $\{ \mpcserver{0}, \mpcserver{1}, \mpcserver{2} \}$ \\
			\midrule
			\Block[r]{8-1}{\rotatebox{90}{Simulations}} 
			&$\psim$                                & Simulation parameters defined by $\ri$\\
			&$\numberofsimulations$    & \# distinct simulations (executed in parallel)\\
			&$\simsteps$                         & \# steps per simulation\\		
			&$\infclass$                           & Infection classes; $\infclass =\{\infclassval{1},\ldots,\infclassval{\numinfclass}\}$\\
			&$\infectionclass_i^{\currstep}$          & $\participant{i}$'s infection class in simulation step $\currstep \in [0,\simsteps]$\\
			&$\encounterset{i}$              & Encounter tokens of $\participant{i}$\\
			& $\maxencounters{i}$         & \# max. encounters by $\participant{i}$ in pre-defined time interval\\
            &$\avgencounters{}$& Average number of encounters\\
			\midrule
			\Block[r]{7-1}{\rotatebox{90}{Protocols}} 
			&$\kappa$        & Computational security parameter $\kappa = 128$ \\
			&$\token_e$                          & Unique token for encounter $e \in [0,\maxencounters{}]$ \\
			&$\likelihood_i^{\token_e}$  & $\participant{i}$'s infection likelihood w.r.t token $\token_e$\\
			& $\Delta_i$    & $\participant{i}$'s cumulative infection likelihood\\
			&$\metadata_i^e$ & Metadata of an encounter $e$ by $\user_i$\\
			&$(\pkey_i, \skey_i)$ & $\participant{i}$'s public/private key pair\\
			&$\signature_i^e$. & $\participant{i}$'s signature on message about encounter $e$ \\
			\bottomrule
		\end{NiceTabular}
	}
	\vspace*{-3pt}
	\caption{Notations.}
	\label{tab:parameters}
    \vspace{-2mm}
\end{table}

%% file: fig/boxes/ripple.tex
\begin{protocolbox}{\ourname{}}{\ourname{} Framework (for one simulation setting).}{fig:sim-framework}
    	\justify
    	\algoHead{\bfcircled{1} - Token Generation} 
    	\begin{itemize}
    		\item $\participant{i} \in \participantset$ executes $\funcCollection$ all the time (on its mobile device), collecting encounter data of the form $(\token_e, \metadata_e)$ with $e < \maxencounters{i}$. For an encounter~$e$, $\token_e$ is the token and $\metadata_e$ is the metadata.
    	\end{itemize}
    	\algoHead{\bfcircled{2} - Simulation Initialization}
    	\begin{itemize}
    		\item $\participant{i} \in \participantset$ \iffullversion receives \else gets \fi $\psim$ from $\ri$ and \iffullversion locally\fi sets $\infectionclass_i^1 = \initialinfectionclass_i$.
    	\end{itemize}
    	\algoHead{\bfcircled{3} - Simulation Execution} 
    	\newline
    	For each simulation step ${\currstep} \in [\simsteps]$, $\participant{i} \in \participantset$ execute:
    	\begin{itemize}
    		\item Filter encounters using $\psim$ to obtain encounter set $\encounterset{i}^{\currstep}$.
    		\item For each token $\token_e \in \encounterset{i}^{\currstep}$, compute the infection likelihood $\likelihood_i^{\token_e}$ locally using the formula from $\ri$.
    		\item Invoke $\funcSimulation$ with the input $\{\likelihood_i^{\token_e}\}_{\token_e \in \encounterset{i}^{\currstep}}$ and obtain $\clikelihood_i^{\currstep} = \sum\limits_{\token_e \in \encounterset{i}^{\currstep}} \rlikelihood_i^{\token_e}$.
    		\item Update the infection class $\infectionclass_i^{\currstep}$ using $\clikelihood_i^{\currstep}$ and guidelines from~$\ri$.
    	\end{itemize}
    	\algoHead{\bfcircled{4} - Result Aggregation}
    	\newline
    	For each simulation step ${\currstep} \in [\simsteps]$, execute:
    	\begin{itemize}
    		\item $\participant{i} \in \participantset$ prepares $\{\simclassval{1}{i},\ldots,\simclassval{\numinfclass}{i}\}^{\currstep}$ with $\simclassval{k}{i} = 1$ if $\infectionclass_i^{\currstep} = \infclassval{k}$ and $\simclassval{k}{i}=0$ otherwise, for $k \in [\numinfclass]$. 
    		\item Invoke $\funcAggregation$ with inputs $\{\simclassval{1}{i},\ldots,\simclassval{\numinfclass}{i}\}^{\currstep}$ to enable $\ri$ obtain the tuple $\{\aggclassval{1},\ldots,\aggclassval{\numinfclass}\}^{\currstep}$, where $\aggclassval{k} = \sum\limits_{\participant{i} \in \participantset} \simclassval{k}{i}$ for $k \in [\numinfclass]$. 
    	\end{itemize}
    \end{protocolbox}

%% file: tex/4_instantiations.tex
We propose two instantiations of $\funcSimulation$ that cover different use cases and offer different trust-efficiency trade-offs. Our first design, \teename{}, is presented in the full version~\cite[\S4.1]{ARXIV:RIPPLE} and assumes the presence of trusted execution environments (TEEs) such as ARM TrustZone on the mobile devices of the participants. In our second design, \pirname{}~(\sect{sec:pirname}), we eliminate this assumption and provide privacy guarantees using cryptographic techniques such as PIR and anonymous communications.

\iffullversion
\subsection{\teename}
\label{sec:teename}
The deployment of the entire operation in a single designated TEE would be a simple solution to achieving the ideal functionality $\funcSimulation$. However, given the massive amount of data that must be handled in a large-scale simulation with potentially millions of users, TEE resource limitations are a prohibitive factor. Furthermore, since the TEE would contain the entire population's contact graph, it would be a single point of failure and an appealing target for an attack on TEE's known vulnerabilities. \teename~(\fig{fig:tee_name}), on the other hand, leverages the presence of TEEs in participants' mobile devices but in a decentralised manner, ensuring that each TEE handles only information related to the encounters made by the respective participant.

Before going into the details of \teename, we will go over the $\funcAnon$ functionality (cf.~\sect{app:anon-comm-channel}), which allows two participants, $\participant{i}$ and $\participant{j}$, to send messages to each other anonymously via a set of communication servers $\mpcservers{}$. The set $\mpcservers{}$ consists of one server acting as an entry node ($\entrynode$), receiving messages from senders, and one server acting as an exit node ($\exitnode$), forwarding messages to receivers. In $\funcAnon$, sender $\participant{i}$ does not learn to whom the message is sent, and receiver $\participant{j}$ does not learn who sent it. Similarly, the servers in $\mpcservers{}$ will be unable to relate receiver and sender of a message. Anonymous communication (cf.~\sect{sec:anonycom}) is an active research area, e.g.,~\cite{alexopoulos2017mcmix,haines2020,eskandarian2021clarion,abraham2020}, and $\funcAnon$ in \teename{} can be instantiated using any of these efficient techniques.

\smallskip
\begin{figure}[htb!]
	\centering
	\includegraphics[width=0.95\columnwidth]{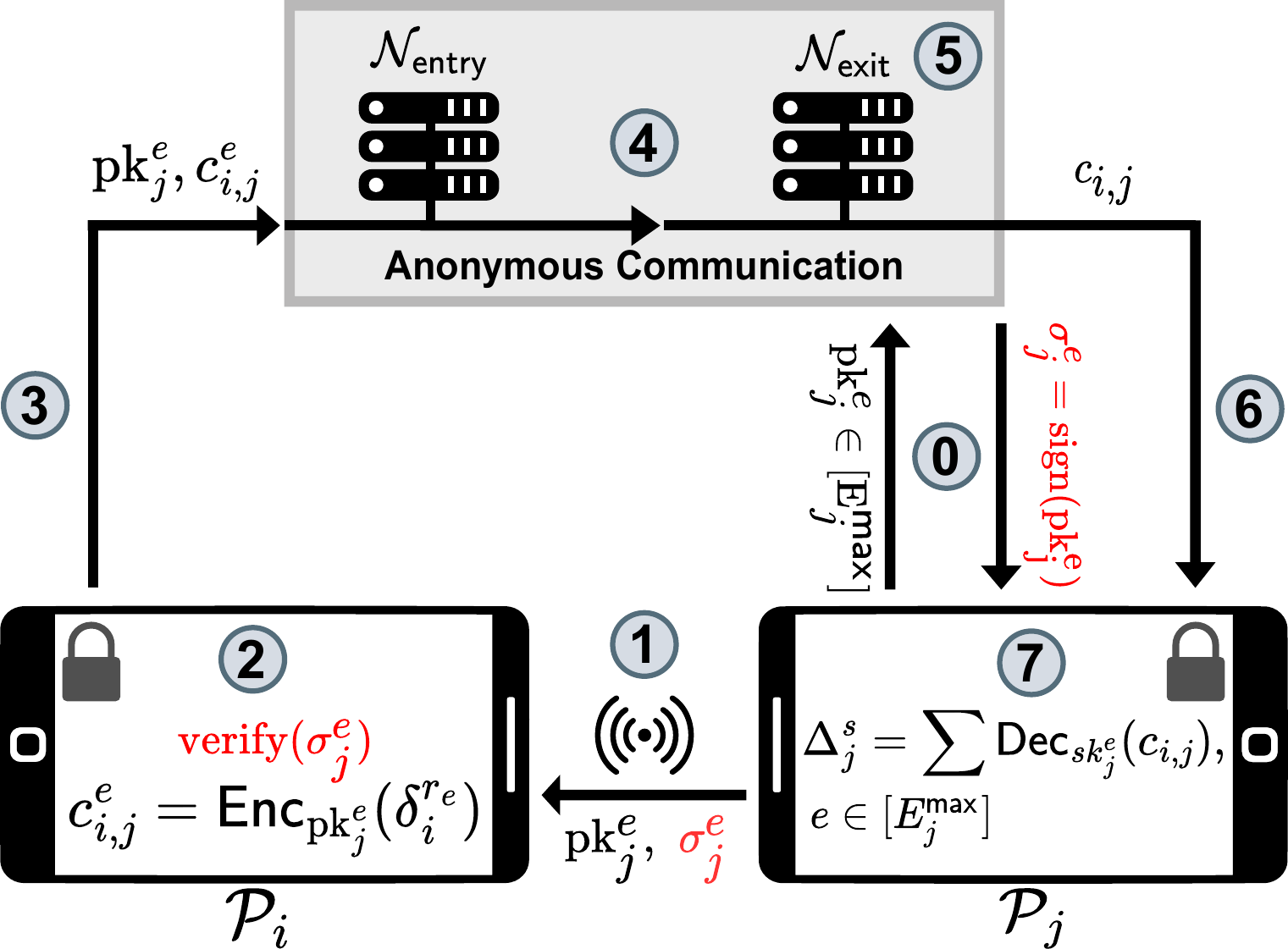}
	\caption{\teename{} Overview. Messages in \textcolor{red}{red} denote additional steps needed for malicious participants.}%
	\label{fig:tee_name}
\end{figure}

\paragraph{Token Generation} (steps \bfcircled{0} to \bfcircled{1} in \fig{fig:tee_name}): During the pre-computation phase, the TEE of each participant $\participant{i} \in \participantset$ generates a list of fresh unique public/private key pairs $(\pkey_i^e, \skey_i^e)$ for all possible encounters $e \in [\maxencounters{i}]$. The keys, for example, can be generated and stored a day ahead of time. The newly generated public keys are then sent by $\participant{i}$'s TEE to the exit node $\exitnode$ (step~\bfcircled{0} in \fig{fig:tee_name}) to enable anonymous communication~(cf.~\sect{app:anon-comm-channel}) via $\funcAnon$ later in the protocol's simulation part. 

During a physical encounter $e$, $\participant{i}$ and $\participant{j}$ exchange two unused public keys $\pkey_i^e$ and $\pkey_j^e$ (step~\bfcircled{1} in \fig{fig:tee_name}). Simultaneously, both participants compute and record metadata $\metadata_e$, such as the time, location, and duration of the encounter, and store this information alongside the received public key.

Additional measures are required for malicious participants to ensure that the participants are exchanging public keys generated by the TEEs: After obtaining the new public keys from $\participant{i}$, the exit node $\exitnode$ goes one step further: It signs them and returns the signatures to $\participant{i}$ after checking that it is connecting directly with a non-corrupted TEE~(step~\bfcircled{0} in \fig{fig:tee_name} and \sect{app:related-primitives}). During a physical encounter, $\participant{j}$ will provide the corresponding signature, denoted by $\sigma_j^e$ along with $\pkey_j^e$ so that the receiver $\participant{i}$ can verify that the key was correctly generated by $\participant{j}$'s TEE (step~\bfcircled{2} in \fig{fig:tee_name}).

\paragraph{Simulation Execution} (steps~\bfcircled{2} to~\bfcircled{7} in \fig{fig:tee_name}): All local computations, including infection likelihood calculation and infection class updates, will be performed within the participants' TEEs. In detail, for each encounter $e$ involving participants $\participant{i}$ and $\participant{j}$, the following steps are executed:
\begin{itemize}
	\item[--] $\participant{i}$'s TEE computes $\likelihood_i^{\token_e}$ and encrypts it using the public key $\pkey_j^e$ of $\participant{j}$ obtained during the token generation phase. Let the ciphertext be $c^e_{i,j} = \Enc{\pkey_j^e}(\likelihood_i^{\token_e})$ (step~\bfcircled{2} in \fig{fig:tee_name}).
	\item[--] $\participant{i}$'s TEE establishes a secure channel with the entry node $\entrynode{}$ of $\mpcservers{}$ via remote attestation and uploads the tuple $(\pkey_j^e, c^e_{i,j})$ (step~\bfcircled{3} in \fig{fig:tee_name}).
	\item[--] The tuple $(\pkey_j^e,c^e_{i,j})$ traverses through the servers in $\mpcservers{}$ and reaches the exit node $\exitnode$ (step~\bfcircled{4} in \fig{fig:tee_name}, instantiation details for the anonymous communication channel are given in~\sect{app:anon-comm-channel}).
	\item[--] If the public key $\pkey_i^e$ has already been used in this simulation step\footnote{This step is not required for semi-honest participants.}, $\exitnode$ discards the tuple (step~\bfcircled{5} in \fig{fig:tee_name}).
	\item[--] Otherwise, $\exitnode$ uses $\pkey_j^e$ to identify the recipient $\participant{j}$ and sends the ciphertext $c^e_{i,j}$ to $\participant{j}$ (step~\bfcircled{6} in \fig{fig:tee_name}).
\end{itemize}

After receiving the ciphertexts for all of the encounters, $\participant{j}$'s TEE decrypts them and aggregates the likelihoods to produce the desired output (step~\bfcircled{7} in \fig{fig:tee_name}).

\subsubsection{Security of \teename{}.}
\label{sec:teename-security}
First, we consider the case of semi-honest participants. 
During the token generation phase, since the current architecture in most mobile devices does not allow direct communication with a TEE while working with Bluetooth LE interfaces, participant $\participant{i}$ can access both the sent and received public keys before they are processed in the TEE. However, unique keys are generated per encounter and do not reveal anything about an encounter's identities due to the security of the underlying $\funcCollection$ functionality, which captures the goal of several contact tracing apps in use. 

The $\funcAnon$ functionality, which implements an anonymous communication channel utilising the servers in $\mpcservers{}$, aids in achieving \emph{contact graph privacy} by preventing participants from learning to/from whom they are sending/receiving messages. While the entry node learns who sends a message, it does not learn who receives them. Similarly, the exit node $\exitnode$ has no knowledge of the sender but learns the recipient using the public key. 
Regarding \emph{confidentiality}, participants in \teename{} have no knowledge of the messages being communicated because they cannot access the content of the TEEs and the TEEs communicate directly to the anonymous channel. Furthermore, servers in $\mpcservers{}$ will not have access to the messages as they are encrypted.

For the case of malicious participants, they could send specifically crafted keys during the token generation phase instead of the ones created by their TEE. However, this will make the signature verification fail and the encounter will get discarded. Furthermore, a malicious participant may reuse public keys for multiple encounters. This manipulation, however, will be useless because the exit node $\exitnode$ checks that each key is only used once before forwarding messages to participants. During the simulation phase, all data and computation are handled directly inside the TEEs of the participants, so no manipulation is possible other than cutting the network connection, i.e., dropping out of the simulation, ensuring \emph{correctness}. Dropouts occur naturally when working with mobile devices and have no effect on privacy guarantees.

\fi

\subsection{\pirname}
\label{sec:pirname}
\iffullversion
In the following, we show how to get rid of \teename's assumption of each participant having a TEE on their mobile devices. If we simply remove the TEE part of \teename{} and run the same protocol, decryption and aggregation of a participant's received infection likelihoods would be under their control. Thus, the individual infection likelihoods of all encounters would be known to them, leaking information about the contact graph (cf.~\sect{sec:privacy-framework}). 
\else
The main challenge for instantiating $\funcSimulation$ is to prevent the leakage of individual infection likelihoods of the encounters.
\fi
To get around this privacy issue, we need to find a way to aggregate the infection likelihoods so that the participants can only derive the sum, not individual values.

\begin{systembox}{$\funcPirSum$}{Ideal functionality for PIR-SUM (semi-honest).}{fig:funcpirsum}
	\justify
	$\funcPirSum$ interacts with $\numservers$ servers, denoted by $\mpcservers{}$, and participant~$\participant{i} \in \participantset$.
	\begin{description}
		\item[Input:] $\funcPirSum$ receives $\threshold$ indices denoted by $\querySet = \{\query_1, \dots, \query_{\threshold}\}$ from $\participant{i}$ and a database $\db$ from $\mpcservers{}$.
		\item[Output:] $\funcPirSum$ sends $\sum_{j=1}^{\threshold} \dbV{\query_j}$ to $\participant{i}$ as the output.
	\end{description}
\end{systembox}
\medskip

Private Information Retrieval~(PIR\iffullversion, cf.~\sect{app:related-primitives}\fi) is one promising solution for allowing participants to retrieve infection likelihoods sent to them anonymously. PIR enables the private download of an item from a public database $\db$ held by $\numservers$ servers without leaking any information to the servers, such as which item is queried or the content of the queried item. However, classical PIR is unsuitable for our needs because we need to retrieve the sum of $\threshold$ items from the database rather than the individual ones. As a result, we introduce the ideal functionality $\funcPirSum$~(\boxref{fig:funcpirsum}), which is similar to a conventional PIR functionality but returns the sum of $\threshold$ queried locations of the database as a result. For the remainder of this section, we consider $\funcPirSum$ to be an ideal black-box and will discuss concrete instantiations in \sect{sec:pirsum}.

\begin{figure}[htb!]
	\centering
	\includegraphics[width=0.95\columnwidth]{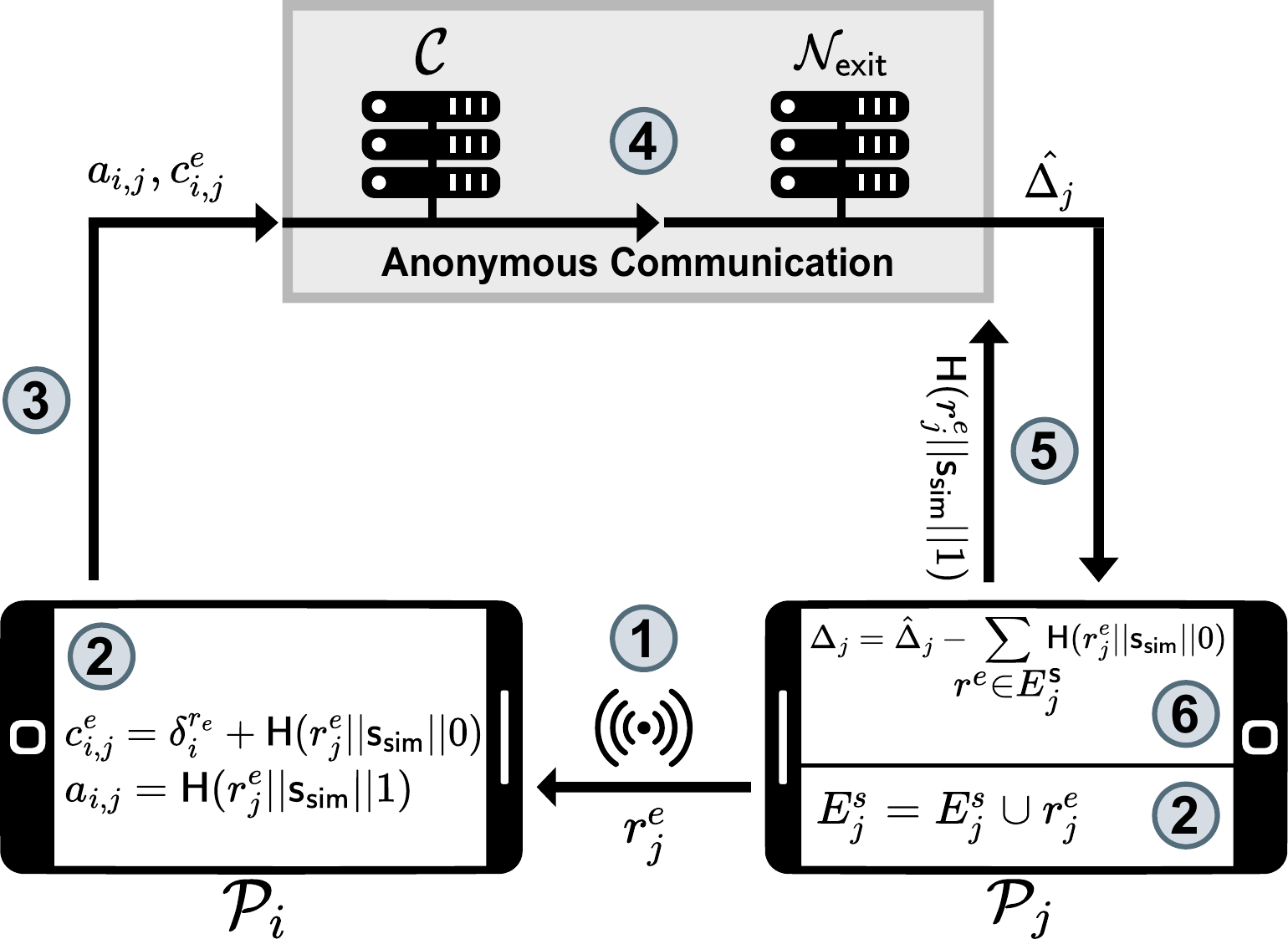}
    \vspace{-1mm}
	\caption{\pirname{} Overview.}%
	\label{fig:pir_name}
    \vspace{-5mm}
\end{figure}

\iffullversion
We now detail the changes needed in the token generation phase to make it compatible with the rest of the \pirname{} protocol.
\else
We now detail the protocol steps of \pirname{} with respect to the overview provided in~\fig{fig:pir_name}.
\fi

\smallskip
\emph{Token Generation (step~\bfcircled{1} in \fig{fig:pir_name}):} During a physical encounter $e$ among participants $\participant{i}$ and $\participant{j}$, they generate and exchange unique random tokens denoted by $\token^e_i$ and $\token^e_j$. 
\iffullversion
Both participants, like in \teename, also record the metadata $\metadata^e$. 
\else
Simultaneously, both participants compute and record metadata $\metadata_e$, such as the time, location, and duration of the encounter, and store this information.
\fi
Thus, at the end of a simulation step ${\currstep} \in [\simsteps]$~(e.g., a day), $\participant{i}$ holds a list of sent encounter tokens, denoted by $\encounterlist^{\currstep}_i = \{\token^e_i\}_{e \in \encounterset{i}}$, where $\encounterset{i}$ is the complete (sent/received) set of encounters of $\participant{i}$, and a list of received tokens, denoted by $\rencounterlist^{\currstep}_i = \{\token^e_j\}_{e \in \encounterset{i}}$. Looking ahead, these random tokens will be used as addresses to communicate the corresponding infection likelihood among the participants.

\smallskip\noindent
\emph{Simulation Execution (steps~\bfcircled{2} to~\bfcircled{6} in \fig{fig:pir_name}):} 
\iffullversion
Local computations such as encounter filtering and infection likelihood calculation proceed similarly to \teename{} but without TEE protection. The steps for an encounter $e$ among $\participant{i}$ and $\participant{j}$ are as follows:
\fi
\begin{acm_item_text}
	\item $\participant{i}$ blinds each infection likelihood $\likelihood_i^{\token_e}$ computed with the corresponding random token $\token^e_j$ received from $\participant{j}$ and obtains the ciphertext $c^e_{i,j} = \likelihood_i^{\token_e} + \Hash(\token^e_j||\simcurrstep||0)$. In addition, it computes the destination address for the ciphertext as $a_{i,j} = \Hash(\token^e_j||\simcurrstep||1)$. Here, $\Hash()$ is a cryptographic hash function and $\simcurrstep \in [\numberofsimulations]$ denotes the current simulation setting. (step~\bfcircled{2} in \fig{fig:pir_name})
	\begin{acm_inneritem}
		\item[--] $\simcurrstep$ is used in $\Hash()$ to ensure that distinct (ciphertext, address) tuples are generated for the same encounters across multiple simulation settings, preventing the exit node $\exitnode$ from potentially linking messages from different simulations.
	\end{acm_inneritem}
	\item $\participant{i}$ sends the tuple $(c^e_{i,j},a_{i,j})$ anonymously to $\exitnode$ with the help of the servers in $\mpcservers{}$. $\exitnode$ discards all the tuples with the same address field ($a_{i,j}$) (steps~\bfcircled{3} to~\bfcircled{4} in \fig{fig:pir_name}). The instantiation details for the anonymous communication channel are given in \iffullversion \sect{app:anon-comm-channel}).\else the full version~\cite[\S B.3]{ARXIV:RIPPLE}).\fi
\end{acm_item_text}

\makeparafit
As a server in $\mpcservers{}$, $\exitnode$ locally creates the database $\db$ for the current simulation step using all of the $(a_{i,j}, c^e_{i,j})$ tuples received (part of step~\bfcircled{4} in \fig{fig:pir_name}). A naive solution of inserting $c^e_{i,j}$ using a simple hashing of the address $a_{i,j}$ will not provide an efficient solution in our case since we require only one message to be stored in each database entry to have an injective mapping between addresses and messages. This is required for the receiver to download the messages sent to them precisely. Simple hashing would translate to a large database size to ensure a negligible probability of collisions. Instead, in \pirname, we use a novel variant of a garbled cuckoo table that we call arithmetic garbled cuckoo table~(AGCT, see below), with $a_{i,j}$ as the insertion key for the database. 

Once the database $\db$ is created, $\exitnode$ sends it to the other servers in $\mpcservers{}$ based on the instantiation of $\funcPirSum$ (cf.~\sect{sec:pirsum}). Each $\participant{j} \in \participantset$ will then participate in an instance of $\funcPirSum$ with the servers in $\mpcservers{}$ acting as PIR servers holding the database $\db$. 
$\participant{j}$ uses the addresses of all its sent encounters from $\encounterlist_{j}^\currstep$, namely $\Hash(\token^e||\simcurrstep||1)$, as the input to $\funcPirSum$ and obtains a blinded version of the cumulative infection likelihood, denoted by $\blindedclikelihood_j$, as the output (step~\bfcircled{5} in~\fig{fig:pir_name}). The cumulative infection likelihood, $\clikelihood_j$, is then unblinded as 
\begin{equation*}
	\clikelihood_j = \blindedclikelihood_j - \sum_{\token^e \in ~\encounterlist_{j}^\currstep } \Hash(\token^e||\simcurrstep||0)
\end{equation*}
concluding the current simulation step (step~\bfcircled{6} in~\fig{fig:pir_name}).

\subsubsection{Security of \pirname{}}
\label{sec:pirname-security}

\iffullversion
Except for the database constructions at exit node $\exitnode$ and the subsequent invocation of the $\funcPirSum$ functionality for the cumulative infection likelihood computation, the security guarantees for semi-honest participants in \pirname{} are similar to those of \teename. 
Unlike \teename, $\exitnode$ in \pirname{} will be unable to identify the message's destination from the address because it will be known only to the receiving participant. Furthermore, each participant obtains the cumulative infection likelihood directly via the $\funcPirSum$ functionality, ensuring that $\exitnode$ cannot infer the participant's encounter details and, thus, \emph{contact graph privacy}.
\else
For semi-honest participants, the $\funcAnon$ functionality, which implements an anonymous communication channel utilising the servers in $\mpcservers{}$, aids in achieving \emph{contact graph privacy} by preventing participants from learning to/from whom they are sending/receiving messages. While the entry node learns who sends a message, it does not learn who receives it. Similarly, the exit node cannot identify the message's destination from the address because it will be known only to the receiving participant. Furthermore, each participant obtains the cumulative infection likelihood directly via the $\funcPirSum$ functionality, ensuring that $\exitnode$ cannot infer the participant's encounter details and, thus, \emph{contact graph privacy} is achieved.
\fi

Malicious participants \iffullversion in \pirname, as opposed to \teename,\fi can tamper with the protocol's correctness by providing incorrect inputs. However, as stated in the threat model in~\sect{sec:framework}, we assume that malicious participants in our framework will not tamper with the correctness and will only seek additional information. A malicious participant could re-use the same encounter token for multiple encounters during token generation, causing the protocol to generate multiple tuples with the same address. However, as the protocol states, $\exitnode$ will discard all such tuples, removing the malicious participant from the system.
Another potential information leakage caused by the participant's aforementioned action is that the entry point of the anonymous communication channel can deduce that multiple participants encountered the same participant. This is not an issue in our protocol because we instantiate the $\funcAnon$ functionality using a 3-server oblivious shuffling scheme\iffullversion~(cf.~\sect{app:anon-comm-channel})\fi, where all the servers except $\exitnode$ will not see any messages in the clear, but only see secret shares.

\subsubsection{Arithmetic Garbled Cuckoo Table (AGCT)}
\label{app:agct}
We design a variant of garbled cuckoo tables \iffullversion(\cite{pinkas2020psi}, cf.~\sect{app:related-primitives})\else\cite{pinkas2020psi} \fi that we term arithmetic garbled cuckoo table (AGCT) to reduce the size of the PIR database while ensuring a negligible collision likelihood. It uses arithmetic sharing instead of XOR-sharing to share database entries and present the details next.

\iffullversion
\medskip
\begin{figure*}[htb!]
	\centering
	\includegraphics[clip, trim=0pt 0pt 0pt 0pt, width=0.8\textwidth]{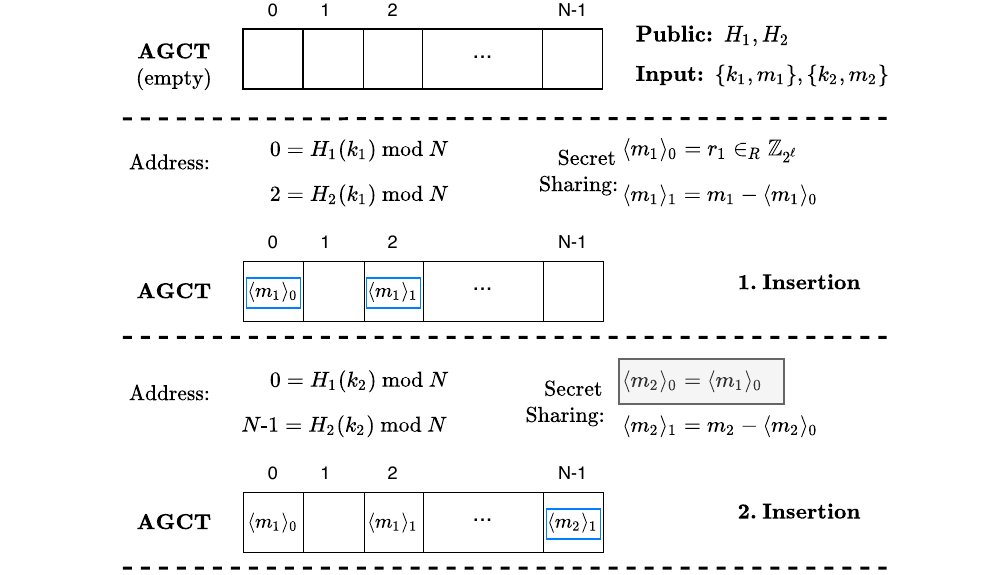}
	\caption{Insertion into the Arithmetic Garbled Cuckoo Table (AGCT). $\Hash_1$ and $\Hash_2$ are two hash functions. $\{\pirkey_1,m_1\}$ and $\{\pirkey_2,m_2\}$ are key-value pairs where the key is used to determine the data address in the database.}%
	\label{fig:agct}
\end{figure*}
\fi

\makeparafit
Assume two key-message pairs $\{\pirkey_1,m_1\}$ and $\{\pirkey_2,m_2\}$\footnote{$k$ corresponds to a key and $m$ to a message in our application.} are to be added to database $D$ with $\dbsize{}$ bins, using two hash function $\Hash_1$ and $\Hash_2$ to determine the insertion addresses. 
\begin{acm_text_list}
    \item[1.] Insertion of $\{\pirkey_1,m_1\}$: 
    \begin{acm_text_sublist}
        \item[a)] Compute $a_1 = \Hash_1(\pirkey_1)\mod{} \dbsize{}$ and $a_2 = \Hash_2(\pirkey_1) \mod{} \dbsize{}$. 
        \item[b)] Check if bins $a_1$ and $a_2$ are already occupied. Let's assume this is not the case.
        \item[c)] Compute the arithmetic sharing of the message $m_1$: $\ashr{ m_1}{0}=r_1\in_R \mathbb{Z}_{2^\ell}$ and $\ashr{m_1}{1}= m_1 - \ashr{ m_1}{0}$.
        \item[d)] Insert $D[a_1]=\ashr{m_1}{0}$ and $D[a_2]= \ashr{m_1}{1}$.
    \end{acm_text_sublist}
    \item[2.] Insertion of $\{\pirkey_2,m_2\}$:
    \begin{acm_text_sublist}
        \item[a)] Compute $b_1 = \Hash_1(\pirkey_2) \mod{} \dbsize{}$ and $b_2 = \Hash_2(\pirkey_2) \mod{} \dbsize{}$. 
        \item[b)] Check if bins $b_1$ and $b_2$ are already occupied. Let's assume $b_1=a_1$, i.e., the first bin is already occupied, but bin $b_2$ is free.
        \item[c)] Compute the arithmetic sharing $m_2$ with $\ashr{m_2}{0}= \ashr{m_1}{0}$ as $b_1=a_1$. Then, the other share is $\ashr{m_2}{1}= m_2 - \ashr{m_2}{0}$.
        \item[d)] Insert $D[b_1]=\ashr{m_2}{0}$ and $D[b_2]=\ashr{m_2}{1}$.
    \end{acm_text_sublist}
\end{acm_text_list}

\emph{Double Collision}: Now the question is how to handle the insertion of a database entry if both addresses determined by the two hash functions are already occupied. An easy solution is to pick different hash functions s.t. no double collision occurs for all $n$ elements that shall be stored in the database. Alternatively, Pinkas et al.~\cite{pinkas2020psi} demonstrate for a garbled cuckoo table how to extend the database by $d + \lambda$ bins, where $d$ is the upper bound of double collisions and $\lambda$ is an error parameter, such that double collisions occur with a negligible likelihood. For details, please refer to~\cite[§5]{pinkas2020psi}.

%% file: tex/5_pir_sum_protocol.tex
\iffullversion
So far, the discussion has focused on RIPPLE as a generic framework composed of multiple ideal functionalities that could be efficiently instantiated using state-of-the-art privacy-enhancing technologies. 
\fi
In this section, we will use three semi-honest MPC servers to instantiate our novel $\funcPirSum$ functionality~(\boxref{fig:funcpirsum}). In particular, we have three servers $\mpcserver{0}, \mpcserver{1}$, and $\mpcserver{2}$, and we design the $\PIRSUM$ protocol to instantiate the $\funcPirSum$ functionality.

The problem statement in our context is formally defined as follows: Participant $\participant{i} \in \participantset$ has a set of $\threshold$ indices denoted by $\querySet = \{\query_1, \dots, \query_{\threshold}\}$ and wants to retrieve $\pirresult = \sum_{\query \in \querySet} \dbV{\query}$. In this case, $\db$ is a database with $\dbsize$ elements of $\ell$-bits each that is held in the clear by both the servers $\mpcserver{1}$ and $\mpcserver{2}$. The server $\mpcserver{0}$ aids in the computation performed by the servers $\mpcserver{1}$ and $\mpcserver{2}$.
Furthermore, we assume a one-time setup \iffullversion(cf.~\sect{app:shared-key})\fi among the servers and $\participant{i}$ that establishes shared pseudorandom keys among them to facilitate non-interactive generation of random values and, thus, save communication~\cite{ASTRA,ABY2,Araki0OPRT21}.

\subsection{Overview of $\PIRSUM$ protocol}
\label{sec:naive-pir}
At a high level, the idea is to use multiple instances of a standard 2-server PIR functionality~\cite{chor1995private,boyle2016fss}, denoted by $\funcPirS$, and combine the responses to get the sum of the desired blocks as the output. $\maskeddb{\pirmask} = \db + \pirmask$ denotes a modified version of the database $\db$ in which every block is summed with the same $\ell$-bit mask $\pirmask$, i.e., $\maskeddbV{\pirmask}{i} = \dbV{i} + \pirmask$ for $i \in [\dbsize]$. The protocol proceeds as follows:
\begin{itemize}
	\item[--] $\mpcserver{1}$ and $\mpcserver{2}$ non-interactively sample $\threshold$ random mask values $\{\pirmask_1, \ldots, \pirmask_{\threshold}\}$ such that $\sum_{j=1}^{\threshold} \pirmask_j = 0$.
	\item[--] $\mpcserver{1}, \mpcserver{2}$, and $\participant{i}$ execute $\threshold$ instances of $\funcPirS$ in parallel, with servers using $\maskeddb{\pirmask_j}$ as the database and $\participant{i}$ using $\query_j$ as the query for the $j$-th instance for $j \in [\threshold]$. 
	\item[--] Let $\pirresult_j$ denote the result obtained by $\participant{i}$ from the $j$-th $\funcPirS$ instance. $\participant{i}$ locally computes $\sum_{j=1}^{\threshold} \pirresult_j$ to obtain the desired result.
\end{itemize}

The details for instantiating $\funcPirS$ using the standard linear summation PIR approach~\cite{chor1995private} are provided in \iffullversion \sect{app:linear-summation}. \else the full version~\cite[\S C]{ARXIV:RIPPLE}. \fi The approach requires $\participant{i}$ to communicate $\dbsize \cdot \threshold$ bits to the servers, which is further reduced in \pirname{} (cf.~\sect{sec:reduce-comm-pir}).

\paragraph{Malicious participants}
\label{sec:sec-naive-pir}
\iffullversion While it is simple to show that the above solution is adequate for semi-honest participants, malicious participants must be dealt with separately. 
\fi
A malicious participant, for example, could use the same query, say $\query_j$, in all $\threshold$ instances and retrieve only the block corresponding to $\query_j$ by dividing the result by $\threshold$. We present a simple verification scheme over the $\funcPirS$ functionality to prevent these manipulations. 

For malicious participants, we want to ensure that $\participant{i}$ used a distinct vector $\bitvector{}$ \iffullversion (representing a PIR query $\query_j$, cf.~\S\ref{app:linear-summation}) \fi during the $\threshold$ parallel instances. One naive approach is to have $\mpcserver{1}$ and $\mpcserver{2}$ compute the bitwise-OR of all the $\threshold$ bit query vectors $\bitvector{}_{1}, \ldots, \bitvector{}_{\threshold}$, and then run a secure two-party computation protocol to compare the number of ones in the resultant vector to $\threshold$. We use the additional server $\mpcserver{0}$ to optimize this step further. $\mpcserver{1}$ and $\mpcserver{2}$ send randomly shuffled versions of their secret shared bit vectors to $\mpcserver{0}$, who reconstructs the shuffled vectors and performs the verification locally. This approach leaks no information to $\mpcserver{0}$ because it has no information about the underlying database $\db$. The verification procedure is as follows:
\begin{tifs_item_text}
	\item[--] $\mpcserver{1}$ and $\mpcserver{2}$ non-interactively agree on a random permutation, denoted by $\permutation$. 
	\item[--] $\mpcserver{u}$ sends $\permutation([\bitvector{}_{j}]_u)$ to $\mpcserver{0}$, $j \in [\threshold]$, $u \in \{1,2\}$.
	\item[--] $\mpcserver{0}$ locally reconstructs $\permutation(\bitvector{}_{j}) = \permutation([\bitvector{}_{j}]_1) \xor \permutation([\bitvector{}_{j}]_2)$, for $j \in [\threshold]$. If all the $\threshold$ bit vectors are correctly formed and distinct, it sends $\accept$ to $\mpcserver{1}$ and $\mpcserver{2}$. Else, it sends $\abort$.
\end{tifs_item_text}

Note that the verification using $\participant{0}$ will incur a communication of $2 \threshold \dbsize$ bits among the servers. Furthermore, the above verification method can be applied to any instantiation of $\funcPirS$ that generates a boolean sharing of the query bit vector among the PIR servers and computes the response as described above, e.g., the PIR schemes of~\cite{chor1995private,boyle2016fss,BonehBCGI21}.

\subsection{Instantiating $\funcPirSum$}
\label{main:PIRSUM-Protocol} 
The formal protocol for $\PIRSUM$ in the case of malicious participants is provided in \boxref{fig:pir-sum-protocol-main}, assuming the presence of an ideal functionality $\funcPirS$ (as will be discussed in $\hybrid{2}$ below). In $\PIRSUM$, the servers $\mpcserver{1}, \mpcserver{2}$ and the participant $\participant{i}$ run $\threshold$ instances of $\funcPirS$ in parallel, one for each query $\query \in \querySet$. Following the execution, $\participant{i}$ receives $\dbV{\query} + \pirmaskquery_{\query}$ whereas $\mpcserver{u}$ receives $\pirmaskquery_{\query}, [\query]_u$, for $u \in \{1,2\}$ and $\query \in \querySet$. $\participant{i}$ then adds up the received messages to get a masked version of the desired output, i.e, $\sum_{\query \in \querySet} \dbV{\query} + \mathsf{mask}_{\querySet}$ with $\mathsf{mask}_{\querySet} = \sum_{\query \in \querySet} \pirmaskquery_{\query}$. $\mpcserver{1}, \mpcserver{2}$ compute $\mathsf{mask}_{\querySet}$ in the same way.

The protocol could be completed by $\mpcserver{1}$ and $\mpcserver{2}$ sending $\mathsf{mask}_{\querySet}$ to $\participant{i}$, then $\participant{i}$ unmasking its value to obtain the desired output. However, before communicating the mask, the servers must ensure that all queries in $\querySet$ are distinct, as shown in $\funcPirSum$~(\boxref{fig:funcpirsumBmain}). 
For this, $\mpcserver{1}, \mpcserver{2}$ use their share of the queries $\query \in \querySet$ and participate in a secure computation protocol with $\mpcserver{0}$. We capture this with an ideal functionality $\funcVerify$, which takes the secret shares of $\threshold$ values from $\mpcserver{1}$ and $\mpcserver{2}$ and returns $\accept$ to the servers if all of the underlying secrets are distinct. Otherwise, it returns $\abort$.

\begin{protocolbox}{$\PIRSUM$}{$\PIRSUM$ Protocol.}{fig:pir-sum-protocol-main}
	\justify
	\emph{Input(s):} i) $\mpcserver{1}, \mpcserver{2}: \db; |\db| = \dbsize$, ii) $\participant{i}: \querySet = \{\query_1, \dots, \query_{\threshold}\}$, and iii) $\mpcserver{0}: \bot$.\\
	\emph{Output:} $\participant{i}: \pirresult = \sum_{\query \in \querySet} \dbV{\query}$ for distinct queries, else $\pirresult = \bot$.
	\justify
	\algoHead{Computation}
	\begin{tifs_item_prot}
		\item[1.] For each $\query \in \querySet$,
		 \begin{tifs_innerprot}
			\item[a.] $\mpcserver{1}, \mpcserver{2}$ and $\participant{i}$ invoke $\funcPirS$ (cf.~$\hybrid{2}$ in proof of Lemma~\ref{lemma:pirsum}) with the inputs $\db, \query$. 
			\item[b.] Let $\pirmaskquery_{\query}, [\query]_u$ denote the output of $\mpcserver{u}$, for $u \in \{1,2\}$ and $\dbV{\query} + \pirmaskquery_{\query}$ denote the output of $\participant{i}$. 
		\end{tifs_innerprot}
		\item[2.] $\participant{i}$ computes  $\pirresult' = \sum_{\query \in \querySet} (\dbV{\query} + \pirmaskquery_{\query})$, while $\mpcserver{1}, \mpcserver{2}$ computes $\mathsf{mask}_{\querySet} = \sum_{\query \in \querySet} \pirmaskquery_{\query}$.
		\item[3.] $\mpcserver{1}, \mpcserver{2}$ and $\mpcserver{0}$ invokes $\funcVerify$ on the secret shares of queries, denoted by $\{[\query]_u\}_{\query \in \querySet, u \in \{1,2\}}$, to check the distinctness of the queries in $\querySet$.
		\item[4.] If $\funcVerify$ returns $\accept$, $\mpcserver{1}, \mpcserver{2}$ sends $\mathsf{mask}_{\querySet}$ to $\participant{i}$, who computes $\pirresult = \pirresult' - \mathsf{mask}_{\querySet}$. Otherwise, $\abort$.
	\end{tifs_item_prot}
\end{protocolbox}

\subsubsection{Security of $\PIRSUM$ Protocol}
\label{main:PIRSUM} 
\boxref{fig:funcpirsumBmain} presents the ideal functionality for $\PIRSUM$ in the context of malicious participants. In this case, $\funcPirSum$ first checks whether all the queries made by the participant $\participant{i}$ are distinct. If yes, the correct result is sent to $\participant{i}$; otherwise, $\bot$ is sent to $\participant{i}$.

\begin{systembox}{$\funcPirSum$}{PIR-SUM functionality (malicious participants).}{fig:funcpirsumBmain}
	\justify
	$\funcPirSum$ interacts with servers in $\mpcservers$, and participant~$\participant{i} \in \participantset$.
	\begin{description}
		\item[Input:] $\funcPirSum$ receives $\threshold$ indices denoted by $\querySet = \{\query_1, \dots, \query_{\threshold}\}$ from $\participant{i}$ and a database $\db$ from $\mpcservers$.
		\item[Computation:] $\funcPirSum$ sets $y = \sum_{j=1}^{\threshold} \dbV{\query_j}$ if all the queries in $\querySet$ are distinct. Else, it sets $y = \bot$.
		\item[Output:] $\funcPirSum$ sends $y$ to $\participant{i}$.
	\end{description}
\end{systembox}

\begin{lemma}
	\label{lemma:pirsum}
	Protocol $\PIRSUM$~(\boxref{fig:pir-sum-protocol-main}) securely realises the $\funcPirSum$ ideal functionality~(\boxref{fig:funcpirsumBmain}) for the case of malicious participants in the $\{\funcPirS, \funcVerify\}$-hybrid model.
\end{lemma}
\begin{proof}
	The proof follows with a hybrid argument based on the three hybrids $\hybrid{0}$, $\hybrid{1}$, and $\hybrid{2}$ discussed below. Furthermore, any secure three-party protocol can be used to instantiate $\funcVerify$ in \ourname{}.

    We use a standard 2-server PIR functionality, denoted by $\funcPirS$, to instantiate $\funcPirSum$. The guarantees of $\funcPirS$, however, are insufficient to meet the security requirements of $\funcPirSum$, so we modify $\funcPirS$ as a sequence of hybrids, denoted by $\hybrid{}$. The modification is carried out in such a way that for a malicious participant $\participant{i}$, each hybrid is computationally indistinguishable from the one before it. As the first hybrid, $\funcPirS$ is denoted by $\hybrid{0}$.
    
    \medskip
    \noindent $\hybrid{0}$:  Let $\funcPirS$ denote a 2-server PIR ideal functionality for our case, with database holders $\mpcserver{1}$ and $\mpcserver{2}$, and client $\participant{i}$. For a database $\db$ held by $\mpcserver{1}$ and $\mpcserver{2}$ and a query $\query$ held by $\participant{i}$, $\funcPirS$ returns $\dbV{\query}$ to $\participant{i}$, but $\mpcserver{1}$ and $\mpcserver{2}$ receive nothing. 
    
    \medskip
    \noindent $\hybrid{1}$:  We modify $\funcPirS$ so that it returns $\dbV{\query} + \pirmaskquery$ to $\participant{i}$, and $\mpcserver{1}, \mpcserver{2}$ receive $\pirmaskquery$, where $\pirmaskquery$ is a random value from the domain of database block size. In other words, the modification can be thought of as the standard $\funcPirS$ being executed over a database $\maskeddb{\pirmaskquery} = \db + \pirmaskquery$ rather than the actual database $\db$. 
    This modification leaks no additional information regarding the query to the servers because they will receive random masks that are independent of the query $\query$. Furthermore, from the perspective of $\participant{i}$ with no prior knowledge of the database $\db$, $\hybrid{1}$ will be indistinguishable from $\hybrid{0}$ because the values it sees in both cases are from the same distribution. As a result, $\hybrid{0} \approx \hybrid{1}$.
    
    \medskip
    \noindent\makeparafit$\hybrid{2}$: Looking ahead, in $\PIRSUM$, servers $\mpcserver{1}, \mpcserver{2}$ and participant $\participant{i}$ run $\threshold$ instances of $\funcPirS$ in parallel, one for each query $\query \in \querySet$.  As shown in $\funcPirSum$~(\boxref{fig:funcpirsumBmain}), the servers must ensure that all of the queries in $\querySet$ are distinct. For this, we modify $\funcPirS$ in $\hybrid{1}$ to additionally output a secret share of the query $\query$ to each of $\mpcserver{1}$ and $\mpcserver{2}$. Because the servers $\mpcserver{1}$ and $\mpcserver{2}$ are assumed to be non-colluding in our setting, this modification will leak no information about the query $\query$ to either server. Since the output to $\participant{i}$ remains unchanged, $\hybrid{1} \approx \hybrid{2}$ from $\participant{i}$'s perspective.
\end{proof}

\subsection{Reducing participant's communication}
\label{sec:reduce-comm-pir}
$\PIRSUM$ is realized in \pirname{} using two different approaches, each with its own set of trade-offs, with a goal of minimizing the communication at the participant's end. While the first approach, denoted by $\PIRSUMA$~(\boxref{fig:pir-sumA-protocol-main}), sacrifices computation for better communication, the second approach, denoted by $\PIRSUMB$~(\boxref{fig:pir-sumB-protocol}), reduces both the computational and communication overhead of the participant in $\PIRSUMA$ with the help of additional server $\mpcserver{0} \in \mpcservers{}$. 

\smallskip
\subsubsection{$\PIRSUMA$~(\boxref{fig:pir-sumA-protocol-main})} 
\label{sec:pir-sumA}
\makeparafit
In this approach, we instantiate $\funcPirS$ using PIR techniques based on Function Secret Sharing (FSS)~\cite{CGibbsB15,boyle2016fss, BonehBCGI21}. To retrieve the $\query$-th block from the database, $\participant{i}$ uses FSS on a Distributed Point Function (DPF)~\cite{GilboaI14} that evaluates to a $1$ only when the input $\query$ is $1$ and to $0$ otherwise. $\participant{i}$ generates two DPF keys $\pirkey_1$ and $\pirkey_2$ that satisfy the above constraint and sends one key to each of the servers $\mpcserver{1}$ and $\mpcserver{2}$. The servers $\mpcserver{1}$ and $\mpcserver{2}$ can then locally expand their key share to obtain their share for the bit vector $\bitvector{}$ and the rest of the procedure proceeds similarly to the naive linear summation method discussed in \sect{sec:naive-pir}\iffullversion ~(more details on Linear Summation PIR are given in \sect{app:linear-summation})\fi . The key size for a database of size $\dbsize$ records using the optimised DPF construction in~\cite{boyle2016fss} is about $\lambda \log(\dbsize/\lambda)$ bits, where $\lambda = 128$ for an AES-based implementation. \boxref{fig:pir-sumA-protocol-main} provides the formal details of the $\PIRSUMA$ protocol.

\paragraph*{Security}
For semi-honest participants, the security of our method directly reduces to that of the 2-server PIR protocol in~\cite{boyle2016fss}. However, as mentioned in~\cite{boyle2016fss}, a malicious participant could generate incorrect DPF keys, risking the scheme's security and correctness. To prevent this type of misbehavior, Boyle et al.~\cite{boyle2016fss} present a form of DPF called ``verifiable DPF'', which can assure the correctness of the DPF keys created by $\participant{i}$ at the cost of an increased \emph{constant} amount of communication between the servers.

\iffullversion
\medskip
\input{fig/boxes/pirsumA}
\fi

While verifiable DPFs in $\PIRSUMA$ ensure the validity of the $\threshold$ bit vectors generated by $\participant{i}$, they do not ensure that bit vectors $\bitvector{}_{1}, \ldots, \bitvector{}_{\threshold}$ correspond to $\threshold$ distinct locations in the database $\db$. However, we use the correctness guarantee of verifiable DPFs to reduce the communication cost for verification, as discussed in \sect{sec:sec-naive-pir}\iffullversion , \sect{app:dpf}, and~\sect{main:fss-pir-scheme}\fi . In detail, all $\threshold$ bit vectors $\bitvector{}_{1}, \ldots, \bitvector{}_{\threshold}$ (PIR queries) are secret-shared between $\mpcserver{1}$ and $\mpcserver{2}$, each guaranteed to have exactly one $1$ and the rest $0$. To ensure distinctness, $\mpcserver{1}$ and $\mpcserver{2}$ XOR their respective $\threshold$ shares locally to obtain the secret-share of a single vector $\bitvector{}_{c} = \xor_{k=1}^{\threshold} \bitvector{}_{k}$. The challenge is then to check if $\bitvector{}_{c}$ has exactly $\threshold 1$ bits. This can be accomplished by having $\mpcserver{1}$ and $\mpcserver{2}$ agree on a random permutation $\permutation$ and reconstructing $\permutation(\bitvector{}_{c})$ to $\mpcserver{0}$ and allowing $\mpcserver{0}$ to perform the check, as in the naive approach (cf.~\sect{sec:sec-naive-pir}).

\iffullversion\else
\medskip
\input{fig/boxes/pirsumA}
\fi

\paragraph*{Computation Complexity (\#AES operations)}
In $\PIRSUMA$, the participant $\participant{i}$ must perform $4 \cdot \log(\dbsize/\lambda)$ AES operations as part of the key generation algorithm for each of the $\threshold$ instances of $\funcPirS$ over a database of size $\dbsize$, where $\lambda = 128$ for an AES-based implementation. Similarly, $\mpcserver{1}$ and $\mpcserver{2}$ must perform $\log(\dbsize/\lambda)$ AES operations for each of the $\dbsize$ DPF evaluations. We refer to Table 1 in~\cite{boyle2016fss} for more specifics.

\smallskip
\subsubsection{$\PIRSUMB$~(\boxref{fig:pir-sumB-protocol})} 
\label{sec:pir-sumB}
In this approach, we use the server $\mpcserver{0}$ to reduce the computation and communication of the participant $\participant{i}$ in $\PIRSUMA$. The idea is that $\mpcserver{0}$ plays the role of $\participant{i}$ for the PIR protocol in $\PIRSUMA$. However, $\participant{i}$ cannot send its query $\query$ to $\mpcserver{0}$ in clear because it would violate privacy. As a result, $\participant{i}$ selects random values $\qshift, \qshiftcorr{\query} \in [\dbsize]$ such that $\query = \qshift + \qshiftcorr{\query}$. In this case, $\qshift$ is a \emph{shifted version} of the index $\query$, and $\qshiftcorr{}$ is a \emph{shift correction} for $\query$. $\participant{i}$ sends $\qshift$ to $\mpcserver{0}$ and $\qshiftcorr{\query}$ to both $\mpcserver{1}$ and $\mpcserver{2}$. The rest of the computation until output retrieval will now occur solely among the servers.

\begin{protocolbox}{$\PIRSUMB$}{$\PIRSUMB$ Protocol.}{fig:pir-sumB-protocol}
	\justify
	\emph{Input(s):} i) $\mpcserver{1}, \mpcserver{2}: \db; |\db| = \dbsize$, ii) $\participant{i}: \querySet = \{\query_1, \dots, \query_{\threshold}\}$, and iii) $\mpcserver{0}: \bot$.\\
	\emph{Output:} $\participant{i}: \pirresult = \sum_{\query \in \querySet} \dbV{\query}$
	\justify
	\algoHead{Computation}
	$\mpcserver{1}$ and $\mpcserver{2}$ sample $\threshold$ random mask values $\{\pirmask_1, \ldots, \pirmask_{\threshold}\} \in \Z{\ell}^{\threshold}$ such that $\sum_{j=1}^{\threshold} \pirmask_j = 0$. For each $\query \in \querySet$, execute the following: 
	\begin{tifs_item_prot}
		\item[1.] $\mpcserver{1}, \mpcserver{2}$ locally compute $\maskeddb{\pirmask_{\query}} = \db + \pirmask_{\query}$, i.e., $\maskeddbV{\pirmask_{\query}}{j} = \dbV{j} + \pirmask_{\query}$, for $j \in [\dbsize]$.
		\item[2.] $\participant{i}, \mpcserver{1}, \mpcserver{2}$ sample random $\qshiftcorr{\query} \in [\dbsize]$.
		\item[3.] $\participant{i}$ computes and sends $\qshift = \query - \qshiftcorr{\query}$ to $\mpcserver{0}$.
		\item[4.] Servers execute DPF protocol~\cite{boyle2016fss} with $\mpcserver{0}$ as client with input $\qshift_{}$. Server $\mpcserver{u}$ obtains $[\bitvector{}_{\qshift}]_u$ with $\bitV{j}{\qshift} = 1$ for $j = \qshift$ and $\bitV{j}{\qshift} = 0$ for $j \ne \qshift$, for $u \in \{1,2\}$.
		\item[5.] $\mpcserver{u}$ locally applies $\qshiftcorr{u}$ on $[\bitvector{}_{\qshift}]_u$ to generate $[\bitvector{}_{\query}]_u$, for $u \in \{1,2\}$.
	\end{tifs_item_prot}
	\algoHead{Verification}
	Let $\{\bitvector{}_{\query_1},\ldots,\bitvector{}_{\query_{\threshold}}\}$ denote the bit vectors whose XOR-shares are generated during the preceding steps: 
	\begin{tifs_item_prot}
		\item[6.] $\mpcserver{k}$ computes $[\bitvector{}_c]_{k} = \xor_{\query \in \querySet} [\bitvector{}_{\query}]_k$, for $u \in \{1,2\}$.
		\item[7.] $\mpcserver{1}$ and $\mpcserver{2}$ non-interactively agree on random permutation $\permutation$. 
		\item[8.] $\mpcserver{u}$ sends $\permutation([\bitvector{}_{c}]_u)$ to $\mpcserver{0}$, for $u \in \{1,2\}$.
		\item[9.] $\mpcserver{0}$ locally reconstructs $\permutation(\bitvector{}_{c}) = \permutation([\bitvector{}_{\query}]_1) \xor \permutation([\bitvector{}_{\query}]_2)$. It sends $\accept$ to $\mpcserver{1}$ and $\mpcserver{2}$, if $\permutation(\bitvector{}_{c})$ has exactly $\threshold$ ones. Else, it sends $\abort$.
	\end{tifs_item_prot}
	\algoHead{Output Transfer}
	Send $\bot$ to $\participant{i}$ if $\mpcserver{0}$ generated $\abort$ during verification. Otherwise, proceed as follows:
	\begin{tifs_item_prot}
		\item[10.] $\mpcserver{u}$ sends $[y_{\query}]_u = \Xor\limits_{j = 1}^{\dbsize} [\bitV{j}{\query}]_u \maskeddbV{\pirmask_{\query}}{j}$ to $\participant{i}$, for $\query \in \querySet, u \in \{1,2\}$.
		\item[11.] $\participant{i}$ locally computes $\pirresult = \sum_{\query \in \querySet} ([y_{\query}]_1 \xor [y_{\query}]_2)$.
	\end{tifs_item_prot}
\end{protocolbox}

The servers run a DPF instance~\cite{boyle2016fss} with $\mpcserver{0}$ acting as the client and input query $\qshift$. At the end of the computation, $\mpcserver{1}$ and $\mpcserver{2}$ obtain the bit vector $\bitvector{}_{\qshift}$, which corresponds to $\qshift$. However, as discussed in $\PIRSUMA$, the servers require an XOR sharing corresponding to the actual query $\query$ to continue the computation. $\mpcserver{1}$ and $\mpcserver{2}$ do this by using the shift correction value $\qshiftcorr{\query}$ received from $\participant{i}$. Both $\mpcserver{1}$ and $\mpcserver{2}$ will perform a right cyclic shift of their $\bitvector{}_{\qshift}$ shares by $\qshiftcorr{\query}$ units. A negative value for $\qshiftcorr{\query}$ indicates a cyclic shift to the left. 

It is easy to see that the XOR shares obtained after the cyclic shift corresponds to the bit vector $\bitvector{}_{\query}$. To further optimise $\participant{i}$'s communication, $\participant{i}$ and servers $\mpcserver{1}, \mpcserver{2}$ non-interactively generate random shift correction values $\qshiftcorr{\query}$ first using the shared-key setup\iffullversion ~(cf.~\sect{app:shared-key})\fi , and only the corresponding $\qshift$ values are communicated to $\mpcserver{0}$. The rest of the protocol is similar to $\PIRSUMA$, and the formal protocol is shown in \boxref{fig:pir-sumB-protocol}. In terms of malicious participants, $\PIRSUMB$ has an advantage over $\PIRSUMA$ as there is no need to use a verifiable DPF to protect against malicious $\participant{i}$, because the semi-honest server $\mpcserver{0}$ generates the DPF key instead of~$\participant{i}$.

\paragraph*{Improving Verification Costs in $\PIRSUMB$}

A large amount of communication is used in the $\PIRSUMB$ protocol to verify malicious participants. More specifically, in Step 8 of \boxref{fig:pir-sumB-protocol}, $2\dbsize$ bits are communicated towards $\mpcserver{0}$ to ensure the distinctness of the queries made by the participant $\participant{i}$. We note that allowing a small amount of leakage to $\mpcserver{0}$ could improve this communication and is discussed next.

Consider the following modification to the $\PIRSUMB$ protocol. Instead of sampling $\qshiftcorr{\query}$ for each query $\query \in \querySet$ (cf. Step 2 in  \boxref{fig:pir-sumB-protocol}), $\participant{i},\mpcserver{1}$, and $\mpcserver{2}$ sample only one random shift value $\qshiftcorr{}$ and uses it for all $\threshold$ instances. Since the queries must be distinct, $\participant{i}$ is forced to send distinct $\qshift$ values to $\mpcserver{0}$ in Step 3 of \boxref{fig:pir-sumB-protocol}. If not, $\mpcserver{0}$ can send $\abort$ to $\mpcserver{1}$ and $\mpcserver{2}$ at this step, eliminating the need for communication-intensive verification.
The relative distance between the queried indices would be leaked to $\mpcserver{0}$ as a result of this optimization. In concrete terms, if we use the same $\qshiftcorr{}$ value for any two queries $\query_m, \query_j \in \querySet$, then $\query_m - \query_n = \qshift_m - \qshift_n$. Because $\mpcserver{0}$ sees all $\qshift$ values in the clear, it can deduce the relative positioning of $\participant{i}$'s actual queries. However, since $\mpcserver{0}$ has no information about the underlying database $\db$, this leakage may be acceptable for some applications.

\subsubsection{Summary of communication costs} 
\tab{tab:pir_sum_costs} summarises the communication cost for our two \PIRSUM{} approaches for instantiating $\funcPirSum$ over a database of size $\dbsize$ with $\threshold$ PIR queries per client.
\fig{fig:plot-encounter} and \tab{tab:comm_cost_avg_enc} show the concrete communication costs of \pirname{} for varying average numbers of encounters~$\avgencounters{}$, which corresponds to the threshold~$\threshold$ in $\PIRSUM$. Our largest benchmarks with a population size of 10M (corresponding to 10M database blocks, leading to a database size of around 1 GBytes for block size $\ell=128$ bits) with 500 average encounters (corresponding to threshold $\threshold=500$) only require around 500 KBytes communication per participant for $\PIRSUMA$ and less than 20 KBytes for $\PIRSUMB$.

\smallskip
\begin{table*}[htb!]
	\centering
		\begin{NiceTabular}{rrr}
			\toprule
			Stage & $\PIRSUMA$ & $\PIRSUMB$ \\
			\midrule
			$\participant{i}$ to servers in $\mpcservers{}$
			& $2\threshold (\securityparam + 2) \log (\dbsize / \securityparam) + 4\threshold \lambda$ 
			& $\threshold \log \dbsize$\\
			Server to server
			& $0$ 
			& $2\threshold (\securityparam + 2) \log (\dbsize / \securityparam) + 4\threshold \lambda$\\
			Servers in $\mpcservers{}$ to $\participant{i}$ 
			& $\threshold \cdot 2 \ell$ & $\threshold \cdot 2 \ell$\\
			\midrule
			+ Verification (mal.) 
			& $2 \dbsize + 2 + \delta$
			& $2 \dbsize + 2$\\
			\bottomrule
		\end{NiceTabular}
    \medskip
	\caption{Summary of communication costs for $\PIRSUM$ to retrieve sum of $\threshold$ indices from a database with $\dbsize$ elements. $\lambda$~denotes the AES key size ($\lambda = 128$ in~\cite{BonehBCGI21}), $\ell$~denotes the block size in bits ($\ell = 128$ in this work), and $\delta$~denotes the constant involved in the verifiable DPF approach~\cite{boyle2016fss}\iffullversion~(cf.~\sect{main:fss-pir-scheme}).\fi.}
	\label{tab:pir_sum_costs}
\end{table*}

%% file: fig/boxes/pirsumA.tex
\begin{protocolbox}{$\PIRSUMA$}{$\PIRSUMA$ Protocol.}{fig:pir-sumA-protocol-main}
	\justify
	\emph{Input(s):} i) $\mpcserver{1}, \mpcserver{2}: \db; |\db| = \dbsize$, ii) $\participant{i}: \querySet = \{\query_1, \dots, \query_{\threshold}\}$, and iii) $\mpcserver{0}: \bot$.\\
	\emph{Output:} $\participant{i}: \pirresult = \sum_{\query \in \querySet} \dbV{\query}$
	\justify
	\algoHead{Computation}
	$\mpcserver{1}$ and $\mpcserver{2}$ sample $\threshold$ random mask values $\{\pirmask_1, \ldots, \pirmask_{\threshold}\} \in \Z{\ell}^{\threshold}$ such that $\sum_{j=1}^{\threshold} \pirmask_j = 0$. For each $\query \in \querySet$, execute the following:  
	\begin{tifs_item_prot}
	    \item[1.] $\mpcserver{1}, \mpcserver{2}$ locally compute $\maskeddb{\pirmask_{\query}} = \db + \pirmask_{\query}$.
		\item[2.] Execute DPF protocol~\cite{boyle2016fss} (verifiable DPF for malicious participants) with $\participant{i}$ as client with input $\query$. Server $\mpcserver{u}$ obtains $[\bitvector{}_{\query}]_u$ with $\bitV{j}{\query} = 1$ for $j = \query$ and $\bitV{j}{\query} = 0$ for $j \ne \query$, for $u \in \{1,2\}$.
	\end{tifs_item_prot}
	\algoHead{Verification}
     Let $\{\bitvector{}_{\query_1},\ldots,\bitvector{}_{\query_{\threshold}}\}$ denote the bit vectors whose XOR-shares are generated during the preceding steps. 
	\begin{tifs_item_prot}
	    \item[3.] \makeparafit Servers verify correctness of~$\query_j$, $j \in [\threshold]$, by executing the \Ver{} algorithm of the verifiable DPF protocol~\cite{boyle2016fss}\iffullversion ~(cf.~\sect{app:dpf})\fi . It outputs $\accept$ to $\mpcserver{1}$ and $\mpcserver{2}$ if~$\query_j$ has exactly 1 one and $(\dbsize - 1)$ zeroes. Else, it outputs $\abort$.
		\item[4.] $\mpcserver{u}$ computes $[\bitvector{}_c]_{u} = \xor_{\query \in \querySet} [\bitvector{}_{\query}]_u$, for $u \in \{1,2\}$.
		\item[5.] $\mpcserver{1}$ and $\mpcserver{2}$ non-interactively agree on random permutation $\permutation$. 
		\item[6.] $\mpcserver{u}$ sends $\permutation([\bitvector{}_{c}]_u)$ to $\mpcserver{0}$, for $u \in \{1,2\}$.
		\item[7.] $\mpcserver{0}$ locally reconstructs $\permutation(\bitvector{}_{c}) = \permutation([\bitvector{}_{\query}]_1) \xor \permutation([\bitvector{}_{\query}]_2)$. It sends $\accept$ to $\mpcserver{1}$ and $\mpcserver{2}$, if $\permutation(\bitvector{}_{c})$ has exactly $\threshold$ ones. Else, it sends $\abort$.
	\end{tifs_item_prot}
	\algoHead{Output Transfer}
	Send $\bot$ to $\participant{i}$ if verifiable DPF or $\mpcserver{0}$ generated $\abort$ during verification. Otherwise, proceed as follows:
	\begin{tifs_item_prot}
		\item[8.] $\mpcserver{u}$ sends $[y_{\query}]_u = \Xor\limits_{j = 1}^{\dbsize} [\bitV{j}{\query}]_u \maskeddbV{\pirmask_{\query}}{j}$ to $\participant{i}$, for $\query \in \querySet, u \in \{1,2\}$.
		\item[9.] $\participant{i}$ locally computes $\pirresult = \sum_{\query \in \querySet} ([y_{\query}]_1 \xor [y_{\query}]_2)$.
	\end{tifs_item_prot}
\end{protocolbox}

%% file: tex/6_evaluation.tex
In this section, we evaluate and compare the computation and communication efficiency of our \iffullversion two \ourname{} protocols \else \pirname{} protocol \fi presented in~\sect{sec:instan}. A fully-fledged implementation, similar to existing contact tracing apps, would necessitate collaboration with industry partners to develop a real-world scalable system for national deployment. Instead, we carry out a proof-of-concept implementation and provide micro benchmark results for all major building blocks.\footnote{Note that we are not attempting to create the most efficient instantiation. More optimizations will undoubtedly improve efficiency, and our protocols can be heavily parallelized with many servers. Instead, our goal here is to demonstrate the viability of \ourname{} protocols for large-scale deployment.}
We focus on the simulation phase for benchmarking, which is separate from the token generation phase. The simulations can ideally be done overnight while mobile phones are charging and have access to a high-bandwidth WiFi connection. According to studies~\cite{walch2016global,Woollaston}, sleeping habits in various countries provide a time window of several hours each night that can be used for this purpose.

\paragraph{Setup and Parameters} 

We run the benchmarks on the server-side with three servers (two for FSS-PIR and one as a helper server as discussed in \sect{sec:reduce-comm-pir}) with Intel Core i9-7960X CPUs@2.8 GHz and 128 GB RAM connected with 10 Gbit/s LAN and 0.1 s RTT. The client is a Samsung Galaxy S10+ with an Exynos 9820@2.73 GHz and 8GB RAM. 
\iffullversion
As Android does not allow third-party developers to implement applications for Android's TEE Trusty~\cite{androidtrusty}, we use hardware-backed crypto operations already implemented by Android instead. 
\fi
We use the code of \cite{kales2019} to instantiate FSS-PIR.
We implement the AGCT in C++ and follow previous work on cuckoo hashing~\cite{pinkas2018scalable} by using tabulation hashing for the hash functions.

We instantiate our protocols in \ourname{} with $\kappa = 128$ bit security. 
\iffullversion
We use RSA-2048 as the encryption scheme in \teename{} since Android offers a hardware-backed implementation. We omit the overhead of remote attestation for the sake of simplicity. 
\fi
For \pirname, we use the FSS-PIR scheme of~\cite{boyle2016fss,kales2019} as the baseline. The addresses are hashed with SHA-256 and trimmed to $40 - 1 + \log_2(\numparticipants \cdot \avgencounters{})$ bits, where $\numparticipants$ is the number of participants and $\avgencounters{}$ represents the average number of encounters per participant per simulation step. We set $\avgencounters{} = 100$ while benchmarking based on numbers provided by research on epidemiological modeling~\cite{mossong2008social,del2007mixing}. To avoid cycles when inserting $n$ messages into the AGCT~(cf.~\sect{app:agct}), we set its size to $10n$. This can be further improved as discussed in~\sect{app:agct}~\iffullversion\cite{pinkas2020psi,pinkas2018scalable,pinkas2018efficient}\else\cite{pinkas2020psi}\fi.
A typical simulation step corresponds to one day, such that 14 simulation steps can simulate two weeks. 

\vspace{-3mm}
\subsection{Communication Complexity}
\label{sec:bench-comm} 
Here, we look at the communication costs that our protocols incur. To analyse the scalability of our protocols, we consider $\numparticipants$ participants ranging from thousand (1K) to twenty million (20M). 
\tab{tab:comm-summary-simstep} summarises the  communication costs of each participant as well as the communication servers ($\mpcservers$) for one simulation step in a specific simulation. One simulation step includes all protocol steps, beginning with participants locally computing their infection likelihood $\likelihood$ and ending with them obtaining their cumulative infection likelihood $\clikelihood$ for that step.

\input{tex/bench/comm_cost_summary}

\subsubsection{Participant Communication}
\iffullversion
As shown in \tab{tab:comm-summary-simstep}, a participant in \teename{} requires just 16KB of total communication in every simulation step, and this is independent of the population on which the simulation is done. This is because each participant will only send and receive infection likelihood messages related to its encounters. While the value in the table corresponds to an average of 100 encounters ($\avgencounters{} = 100$), we depict the participants' communication in \fig{fig:plot-encounter} with varied $\avgencounters{}$ ranging from 10 to 500 for a population of 10M. Note that a 2-week simulation with $\avgencounters{} = 500$ can be completed by a participant in \teename{} with roughly 1MB of communication. 
\fi

\iffullversion
Unlike \teename, participant communication in both $\PIRSUMA$ and $\PIRSUMB$ increases for larger populations as the corresponding database size increases. The communication, however, is only sub-linear in the database size\footnote{DB size of $10n$ and communication costs of \pirname{} can be reduced by optimizing the database size by extending the database by only $d + \lambda$ bins, where $d$ is the upper bound of double collisions and $\lambda$ is an error parameter (cf.~\sect{app:agct} and \cite{pinkas2020psi}).}.

\else
We depict the participants' communication in \fig{fig:plot-encounter} with varied average number of encounters $\avgencounters{}$ ranging from 10 to 500 for a population of 10M. The participant communication increases sub-linearly with the database size.
\fi
\iffullversion
In particular, the participant's communication in $\PIRSUMA$ ranges from 51.63KB to 98.06KB, with the higher cost over \teename{} attributed to the size of DPF keys used in the underlying FSS-PIR scheme~\cite{boyle2016fss}, as discussed in \sect{sec:pirsum}. 
\else
In particular, the participant's communication in $\PIRSUMA$ ranges from 51.63KB to 98.06KB.
\fi
The communication in $\PIRSUMB$, on the other hand, is about 3.5KB for all participant sizes we consider. This reduced communication is due to the optimization in $\PIRSUMB$, which offloads the DPF key generation task to the helper server $\mpcserver{0}$ (cf.~\sect{sec:pir-sumB}). A participant in $\PIRSUMA$ must communicate approximately 7MB of data for a 2-week simulation for a 10M population with $\avgencounters{} = 500$, whereas it is only 0.25MB in the case of $\PIRSUMB$.

\input{tex/bench/plot_encounter_10M}


\tab{tab:comm_cost_avg_enc} provides the communication per participant for multiple population sizes in \iffullversion \teename{}, $\PIRSUMA$, and $\PIRSUMB$, \else $\PIRSUMA$ and $\PIRSUMB$, \fi while varying the average number of encounters $\avgencounters{}$ per simulation step from 10 to 500. 
\iffullversion
The communication cost in \teename{} is independent of the population size and grows linearly in $\avgencounters{}$. 
A similar trend can be seen in \pirname{} except that the cost increases sub-linearly with the population size due to using the FSS-based PIR scheme in \pirname. 
\else
The communication cost in \pirname{} increases sub-linearly with the population size and linearly in $\avgencounters{}$.
\fi

\subsubsection{Server Communication}
The servers' communication is primarily attributed to the anonymous communication channel they have established, which provides unlinkability and, thus, privacy to the participants' messages. Communicating $M$ messages through the channel requires the servers to communicate \iffullversion $2M$ messages in \teename, and \fi $3M$ messages in \pirname{}\iffullversion\else~(cf. \cite[\S B.3]{ARXIV:RIPPLE})\fi. 
\iffullversion
When it comes to concrete values, however, the server communication in \pirname{} is half that of \teename{}, as shown in \tab{tab:comm-summary-simstep}. This is due to the larger message size in \teename{} due to the use of public-key encryption.
\fi


For a population of 10M, the servers \iffullversion in \teename{} must communicate 192GB of data among themselves, whereas \pirname{} requires 96GB.\else must communicate 96GB of data among themselves.\fi
Setting the proper bit length for the address field in the messages can further reduce communication.  For example, a population of 20M with $\avgencounters{} = 100$ can be accommodated in a 70-bit address field. Using this optimization will result in an additional 23~\% reduction in communication at the servers, as shown in \iffullversion\tab{tab:comm-micro-simstep}. \fig{fig:shuffle-teepir} captures these observations better,  and \tab{tab:comm-micro-simstep} and \tab{tab:comm_cost_avg_enc} in the next subsection provide a detailed analysis of the concrete communication costs.
\else
\fig{fig:shuffle-teepir} and \tab{tab:comm_cost_avg_enc}.
\fi

\input{tex/bench/comm_costs_varying_avg_encounter}

\iffullversion
\subsubsection{Communication Micro Benchmarks}

\tab{tab:comm-micro-simstep} details the communication costs per simulation step at various stages in our instantiations of \ourname. We find that a participant's communication costs are very low compared to the overall costs. In \teename, a participant communicates at most 268 KB and incurs a runtime of 92 seconds over a two-week simulation over a population of one million. In $\PIRSUMB$, the cost is reduced to 100 KB and 40 seconds of runtime. Communication increases to 1.2~MB in $\PIRSUMA$ due to the participant's handling of DPF keys.
 
Finally, \tab{tab:comm-micro-simstep} does not include costs for verification against malicious participants since they can be eliminated using server $\mpcserver{0}$ (cf.~\sect{sec:pir-sumB}) or sketching algorithms similar to those in~\cite{boyle2016fss}.

\input{tex/bench/microbench_communication.tex}
\fi

\vspace{-3mm}
\subsection{Computation Complexity}
\label{sec:bench-runtime} 
\makeparafit
This section focuses on the runtime, including computation time and communication between entities. \tab{tab:participant-computation} summarizes the computation time with respect to a participant $\participant{i}$ for a two-week simulation over a half-million population.
\iffullversion
The longer computation time in \teename{}, as shown in \tab{tab:participant-computation}, is due to the public key encryption and decryption that occurs within the mobile device's TEE. This cost, however, is independent of population size and scales linearly with the average number of encounters, denoted by $\avgencounters{}$. 
In particular, for a 14-day simulation with a population of half a million, $\participant{i}$ in \teename{} needs approximately 43.7 seconds to perform the encryption and decryption tasks and may require additional time for the remote attestation procedure, which is not covered in our benchmarks. 
\fi
$\participant{i}$'s computation time in \pirname\iffullversion, on the other hand, is significantly lower and \fi ~is at most $5$ milliseconds for $\PIRSUMB$, while it increases to around 165 milliseconds for $\PIRSUMA$. The increased computation time in $\PIRSUMA$ is due to DPF key generation, which scales sub-linearly with population size.

\input{tex/bench/comp_cost_participant}

\iffullversion
In \fig{fig:runtime-full}, we plot the overall runtime of our two instantiations in \ourname{} for a full simulation of 2 weeks over various populations ranging from 1K to 500K. After a population of 100K, the runtime of \pirname{} begins to exceed that of \teename{} due to an increase in database size, which results in longer data transfer times.
More details regarding computation time is presented in \tab{tab:comp-micro-simstep}. Note that the runtimes in \fig{fig:runtime-full} include runtime for computation and communication of the secure shuffle among the servers for anonymous communication and among servers and clients for the PIR in \pirname.

\vspace{2mm}
\input{tex/bench/plot_time_simfull}
\else

We measured the overall runtime of \pirname{} for a simulation of 2 weeks (including secure shuffling and anonymous communication). The simulation over a population of $p$=10K takes 74s, and 453s for a population of $p$=500K.

\fi
\iffullversion
\vspace{-2mm}
\subsubsection{Computation Micro Benchmarks}

\tab{tab:comp-micro-simstep} contains the computation costs per simulation at the different stages of our instantiations of \ourname's. As visible, data transfer time as part of anonymous communication through servers accounts for the majority of computation time and begins to affect overall performance as the population grows. Our system crashed due to memory constraints after a population of 500K while running the experiments. This will not be the case in a real-world deployment of powerful servers linked by high-speed networks. Similar as w.r.t. communication, participants' computation costs are very low in comparison to the overall costs.

\input{tex/bench/microbench_computation.tex}

\fi

\subsubsection{Battery Usage} The token generation phase in RIPPLE consumes the most amount of mobile battery as this phase is active throughout the day. However, the main interaction during this phase is exchanging two random tokens, which is similar to existing contact tracing apps. This usage could be optimized by mobile OS providers like Apple and Google, as discussed by Vaudenay et al.~\cite{vaudenay2020analysis} and Avitabile et al.~\cite{Avitabile2020} in the context of contact tracing apps. Their technology enables an app to run in the background, thus, significantly improving battery life, which is otherwise impossible for a standard third-party mobile application. Additionally, \ourname{} could offer users the choice to participate only in simulations while charging so as not to cause any unwanted battery drain.

\subsubsection{Comparison to Related Work} Note that no experimental comparison to related work is (and can be) done, as \ourname{} is the first distributed privacy-preserving epidemiological modeling system. Established contact tracing apps, such as the SwissCovid\footnote{\url{https://github.com/SwissCovid}}, the German Corona-Warn-App\footnote{\url{https://www.coronawarn.app/en/}}, or the Australian COVIDSafe\footnote{\myurl{https://www.health.gov.au/resources/apps-and-tools/covidsafe-app}} only record contacts for notifying contacts of infected people. Concretely, contact tracing basically relates to \ourname{}'s token generation phase, while the other three phases (simulation initialization, simulation execution, and result aggregation, cf.~\sect{sec:phases-framework}) are not covered by any contact tracing system. Crucially, the main contribution of our work is how to realize the simulation execution, which has never been done before. Hence, no meaningful comparison between the systems is possible due to differences in the fundamental functionalities.

\subsubsection{Code availability}
Available at \href{ https://encrypto.de/code/RIPPLE}{\texttt{https://encrypto.de/code/RIPPLE}}.

\iffullversion
\paragraph{Summmary} Our benchmarking using the proof-of-concept implementation demonstrated the \ourname{} framework's viability for real-world adaptation. One of the key benefits of our approaches is that participants have very little work to do. The system's efficiency can be further improved with appropriate hardware and optimized (non-prototype) implementations.
\fi

%% file: tex/bench/comm_cost_summary.tex
\begin{table*}[htb!]
    \centering
	\resizebox{0.98\textwidth}{!}{
		\begin{NiceTabular}{rrrrrrrrrrrr}[notes/para]
			\toprule
			\Block[c]{2-1}{Entities} & \Block[c]{2-1}{Protocol} 
			& \Block[c]{1-10}{Population ($\numparticipants$)} & & & & & & & & &\\
			\cmidrule{3-12}
			& & 1K & 10K & 50K & 100K & 500K & 1M & 2M & 5M & 10M & 20M\\
			\midrule \midrule
			\iffullversion \Block[c]{3-1}{Participants in $\participantset$ \\ (in KB)} 
			& \teename{} ~(\sect{sec:teename})
			& 16.00 & 16.00 & 16.00 & 16.00 & 16.00 & 16.00 & 16.00 & 16.00 & 16.00 & 16.00 \\ \cmidrule{2-12}
			\else \Block[c]{2-1}{Participants in $\participantset$ \\ (in KB)}  \fi & \pirname{}: $\PIRSUMA$ ~(\sect{sec:pir-sumA})
			& 51.63 & 62.42 & 69.97 & 73.22 & 80.77 & 84.02 & 87.27 & 91.56 & 94.81 & 98.06 \\  \cmidrule{2-12}
			& \pirname{}: $\PIRSUMB$ ~(\sect{sec:pir-sumB})
			&  3.45 & 3.49 & 3.52 & 3.53 & 3.56 & 3.57 & 3.59 & 3.60 & 3.62 & 3.63 \\ 
			\midrule \midrule
            \iffullversion \Block[c]{2-1}{Servers in $\mpcservers$ \\ (in GB)} 
            & \teename{} ~(\sect{sec:teename})
            & 0.02 & 0.19 & 0.96 & 1.92 & 9.60 & 19.20  & 38.40 & 96.00 & 192.00 & 384.00 \\  \cmidrule{2-12}
            \else Servers in $\mpcservers$ (in GB) \fi & \pirname{} ~~(\sect{sec:pirsum})
            & 0.01 & 0.10 & 0.48 & 0.96 & 4.80 & 9.60 & 19.20 & 48.00 & 96.00 & 192.00 \\ 
            \bottomrule
		\end{NiceTabular}
	  }
      \caption{Communication costs per simulation step in our \ourname{} instantiations.}
      \label{tab:comm-summary-simstep}
\end{table*}

%% file: tex/bench/plot_encounter_10M.tex
\begin{figure}[htb!]
    \centering
    \begin{minipage}{0.99\columnwidth}
    \centering
	\resizebox{0.85\textwidth}{!}{
    \begin{tikzpicture}
        \begin{axis}[
        axis lines=middle,
        ymin=0,
        x label style={at={(current axis.right of origin)},anchor=north, right=2mm, font=\Large},
        y label style={font=\Large},
        legend style={at={(0.05,0.7)},anchor=west,font=\large},
        legend cell align={right},
        cycle list name=exotic,
        xlabel= $\avgencounters{}$,
        ylabel=Comm. (in KB),
        xticklabel style = {rotate=30,anchor=east,font=\Large},
        yticklabel style = {font=\Large},
        enlargelimits = false,
        xticklabels from table={encountersplot.dat}{Population},xtick=data]
        \iffullversion\addplot+[thick,mark=star,blue] table [y= TEEC,x=X]{encountersplot.dat};
        \addlegendentry{\teename{}}]\fi
        \addplot+[thick,mark=square*,red] table [y= PIRC,x=X]{encountersplot.dat};
        \addlegendentry{\pirname: $\PIRSUMA$}]
        \addplot+[thick,mark=*,darkgreen] table [y= PIR2C,x=X]{encountersplot.dat};
        \addlegendentry{\pirname: $\PIRSUMB$}]
        \end{axis}
    \end{tikzpicture}
	}
    \vspace{-1mm}
    \subcaption{Participant's communication with varying $\avgencounters{}$ for a population of 10M.}\label{fig:plot-encounter}
    \end{minipage}
    \begin{minipage}{0.99\columnwidth}
    \centering
    \resizebox{0.85\textwidth}{!}{
    \begin{tikzpicture}
        \begin{axis}[
        axis lines=middle,
        ymin=0,
        x label style={at={(current axis.right of origin)},anchor=north, above=2mm, font=\Large},
        y label style={font=\Large},
        legend style={at={(0.05,0.7)},anchor=west,font=\large},
        legend cell align={left},
        cycle list name=exotic,
        xlabel=Population ($\numparticipants$),
        ylabel=Comm. (in GB),
        xticklabel style = {rotate=30,anchor=east,font=\Large},
        yticklabel style = {font=\Large},
        enlargelimits = false,
        xticklabels from table={shuffleservers.dat}{Population},xtick=data]
        \iffullversion
        \addplot+[thick,mark=square*] table [y=TEE,x=X]{shuffleservers.dat};
        \addlegendentry{\teename{}}
        \fi
        \addplot+[thick,mark=square*] table [y= PIR,x=X]{shuffleservers.dat};
        \addlegendentry{\pirname{}}]
        \addplot+[thick,mark=square*] table [y= PIRB,x=X]{shuffleservers.dat};
        \addlegendentry{\pirname{}$^\star$}]
        \end{axis}
    \end{tikzpicture}
	}
    \vspace{-1mm}
    \subcaption{Servers' communication per simulation step for varying population. $^\star$ denotes the results for optimized bit addresses in \pirname{}~(cf. full version~\cite{ARXIV:RIPPLE}).}\label{fig:shuffle-teepir}
    \end{minipage}
    \caption{Communication Costs of \ourname.}
    \label{fig:plot-encounter-communication}
\end{figure}

%% file: tex/bench/comm_costs_varying_avg_encounter.tex
\begin{table}[htb!]
	\centering
	\resizebox{\columnwidth}{!}{
		\begin{NiceTabular}{r r r r r r r}
			\toprule 
			\Block[c]{2-1}{Population $\numparticipants$} 
			& \Block[c]{2-1}{Protocol} 
			& \Block[c]{1-5}{$\avgencounters{}$} & & & & \\ \cmidrule{3-7} 
			& & 10 & 50 & 100 & 250 & 500 \\
			\midrule
	       \iffullversion\Block[c]{3-1}{100K} 
	       &\teename{} ~~(\sect{sec:teename}) & 1.60 & 8.00 & 16.00 & 40.00 & 80.00 \\\cmidrule{2-7}
	     \else \Block[c]{2-1}{100K} \fi  &\pirname{}: $\PIRSUMA$ & 6.24 & 34.99 & 73.22 & 193.79 &403.83\\\cmidrule{2-7}
	       &\pirname{}: $\PIRSUMB$& 0.35 & 1.76 & 3.53 & 8.87 & 17.81\\
	       \midrule\midrule
	       \iffullversion\Block[c]{3-1}{1M} 
	       &\teename{} ~~(\sect{sec:teename}) & 1.60 & 8.00 & 16.00 & 40.00 & 80.00\\\cmidrule{2-7}
	    \else\Block[c]{2-1}{1M} \fi   &\pirname{}: $\PIRSUMA$& 7.32 & 40.38 & 84.02 & 220.78 &457.81\\\cmidrule{2-7}
	       & \pirname{}: $\PIRSUMB$ & 0.35 & 1.78 & 3.57 & 8.98 & 18.01\\\midrule\midrule
	       \iffullversion\Block[c]{3-1}{10M} 
	       &\teename{} ~~(\sect{sec:teename}) & 1.60 & 8.00 & 16.00 & 40.00 & 80.00 \\\cmidrule{2-7}
	    \else\Block[c]{2-1}{10M} \fi   &\pirname{}: $\PIRSUMA$ & 8.40 &45.78 & 94.81 & 247.77 & 511.79\\\cmidrule{2-7}
	       & \pirname{}: $\PIRSUMB$ &0.36 & 1.80 & 3.62 & 9.08 & 18.22\\
		   \bottomrule
		\end{NiceTabular}
	}
	\smallskip
	\caption{Communication (in KB) per participant in a simulation step for varying average numbers of encounters $\avgencounters{}$ and population sizes $\numparticipants$.}
	\label{tab:comm_cost_avg_enc}
\end{table}

%% file: tex/bench/microbench_communication.tex
\begin{table*}[t]
	\centering
    \resizebox{\textwidth}{!}{        
        \begin{NiceTabular}{rrrrrrrrrrrr}[notes/para] 
            \toprule
            \Block[c]{2-1}{Stages of \ourname} & \Block[c]{2-1}{Protocol\tabularnote{$\CIRCLE$ - $128$-bit address for \pirname{} and $\LEFTcircle$ - $40 - 1 + \log_2(\numparticipants \cdot \avgencounters{})$ bit address for \pirname{}.}} 
            & \Block[c]{1-10}{Population ($\numparticipants$)} & & & & & & & & &\\
            \cmidrule{3-12}
            & & 1K & 10K & 50K & 100K & 500K & 1M & 2M & 5M & 10M & 20M\\
            \midrule \midrule
            \Block[c]{3-1}{Message Generation\\ by $\participant{i} \in \participantset$ \\ (in KB)} 
            & \teename{} ~~(\sect{sec:teename})\tabularnote{Includes registration of public keys with the exit node $\exitnode$.}
            & 12.80 & 12.80 & 12.80 & 12.80 & 12.80 & 12.80 & 12.80 & 12.80 & 12.80 & 12.80 \\ \cmidrule{2-12}
            & \pirname{}: $\CIRCLE$ ~~(\sect{sec:pirname})& 3.20 & 3.20 & 3.20 & 3.20 & 3.20 & 3.20 & 3.20 & 3.20 & 3.20 & 3.20 \\
            \cmidrule{2-12}
            &\pirname{}: $\LEFTcircle$ ~~(\sect{sec:pirname})& 2.30 & 2.34 & 2.38 & 2.39 & 2.41 & 2.43 & 2.44 & 2.45 & 2.46 & 2.48 \\
            \midrule \midrule
            \Block[c]{3-1}{Secure Shuffle by $\mpcservers$ \\ (in GB)} 
            & \teename{} ~~(\sect{sec:teename})
            & 0.02 & 0.19 & 0.96 & 1.92 & 9.60 & 19.20 & 38.40 & 96.00 & 192.00 & 384.00 \\  \cmidrule{2-12}
            & \pirname{} - $\CIRCLE$ ~~(\sect{sec:pirname})
            & 0.01 & 0.10 & 0.48 & 0.96 & 4.80 & 9.60 & 19.20 & 48.00 & 96.00 & 192.00 \\ 
            \cmidrule{2-12}
            & \pirname{} - $\LEFTcircle$ ~~(\sect{sec:pirname})
            & 0.01 & 0.07 & 0.36 & 0.72 & 3.62 & 7.28 & 14.63 & 36.75 & 73.88 & 148.50 \\ 
            \midrule \midrule
            \Block[c]{4-1}{Output Computation\\ by $\participant{i} \in \participantset$ \\ (in KB)\tabularnote{includes message download, decryption/PIR queries, summation.}} 
            & \teename{} ~~(\sect{sec:teename})
            & 6.40 & 6.40 & 6.40 & 6.40 & 6.40 & 6.40 & 6.40 & 6.40 & 6.40 & 6.40 \\ \cmidrule{2-12}
            & $\PIRSUMA$ - $\CIRCLE$ ~~(\sect{sec:pir-sumA})
            & 51.36 & 62.42 & 69.97 & 73.22 & 80.77 & 84.02 & 87.27 & 91.56 & 94.81 & 98.06 \\
            \cmidrule{2-12}
            & $\PIRSUMA$ - $\LEFTcircle$ ~~(\sect{sec:pir-sumA})
            &26.48 & 32.64 & 37.69 & 39.82 & 44.77 & 47.05 & 49.38 & 52.33 & 54.77 & 57.26\\
            \cmidrule{2-12}
            & $\PIRSUMB$ ~~(\sect{sec:pir-sumB})
            &3.45 & 3.49 & 3.52 & 3.53 & 3.56 & 3.57 & 3.59 & 3.60 & 3.62 & 3.63\\
            \bottomrule
        \end{NiceTabular}
    }
    \caption{Detailed communication costs per simulation step in \ourname{}.}\label{tab:comm-micro-simstep}
\end{table*}

%% file: tex/bench/comp_cost_participant.tex
\begin{table}[htb!]
	\centering
	\resizebox{\columnwidth}{!}{
	\begin{NiceTabular}{r r r r}
		\toprule
		&\Block[c]{1-3}{Per Simulation Step / Simulation ($\simsteps = 14$)} & &\\
		\cmidrule{2-4}
		&\Block[c]{1-1}{\small Message  Generation\\\small (in ms)} 
		&\Block[c]{1-1}{\small PIR  Queries\\\small (in ms)}
		&\Block[c]{1-1}{\small Output   Computation\\\small (in ms)}
		\\
		\midrule\midrule
		\iffullversion\teename{} 
		& 1.12 / 80.00 & -           & 42.56 / 3040.00 \\
		\midrule
        \fi
		$\PIRSUMA$ 
		& 0.30 / 4.26  &11.73 / 160 & 4.8e-2 / 6.72e-1\\
		\midrule
		$\PIRSUMB$ 
		& 0.30 / 4.26 & 3.0e-3 / 4.2e-2 & 4.8e-2 / 6.72e-1\\
		\bottomrule
	\end{NiceTabular}
    }
    \smallskip
	\caption{Average participant computation times per simulation step distributed across various tasks. Values are obtained using a mobile for a population of 500K with $\avgencounters{}=100$.}
    \label{tab:participant-computation}
\end{table}

%% file: tex/bench/plot_time_simfull.tex
\begin{figure}[htb!]
    \centering
	\resizebox{0.4\textwidth}{!}{
		\begin{tikzpicture}
			\begin{axis}[
				axis lines=middle,
				ymin=0,
				x label style={at={(current axis.right of origin)},anchor=north, above=2mm},
				legend style={at={(0.05,0.7)},anchor=west},
				legend cell align={left},
				cycle list name=exotic,
				xlabel=Population ($\numparticipants$),
				ylabel=Time. (in sec),
				xticklabel style = {rotate=30,anchor=east},
				enlargelimits = false,
				xticklabels from table={runtime.dat}{Population},xtick=data]
				\iffullversion\addplot+[thick,mark=*,blue] table [y=TEEF,x=X]{runtime.dat};
				\addlegendentry{\teename{}}\fi
				\addplot+[thick,mark=square*,red] table [y= PIRF,x=X]{runtime.dat};
				\addlegendentry{\pirname{}}]
			\end{axis}
		\end{tikzpicture}
	}
    \captionsetup{font={small}}
	\caption{Runtime per simulation in \ourname{} (14 days).}
	\label{fig:runtime-full}
\end{figure}

%% file: tex/bench/microbench_computation.tex
\begin{table*}[t!]
	\centering
    \begin{center}
    \begin{NiceTabular}{rrrrrrrr}[notes/para]
        \toprule
        \Block[c]{2-1}{Stages of \ourname} & \Block[c]{2-1}{Protocol} 
        & \Block[c]{1-6}{Population ($\numparticipants$)} & & & & & \\
        \cmidrule{3-8}
        & & 1K & 10K & 50K & 100K & 500K & 1M \\
        \midrule \midrule
        \Block[c]{2-1}{Message Generation by $\participant{i} \in \participantset$ \\ (in sec)} 
        & \teename{} ~~(\sect{sec:teename})
        & 1.12 & 1.12 & 1.12 & 1.12 & 1.12 & 1.12\\ \cmidrule{2-8}
        & \pirname{}: ~~(\sect{sec:pirname})
        & 4.26e-3 & 4.26e-3 & 4.26e-3 & 4.26e-3 & 4.26e-3 & 4.26e-3\\
        \midrule \midrule
        \Block[c]{2-1}{Secure Shuffle by $\mpcservers$ \\ (in sec)} 
        & \teename{} ~~(\sect{sec:teename})
        & 0.70 & 5.20 & 25.38 & 60.77 & 211.47  & 493.33$^{\textcolor{red}{\star}}$\tabularnote{$^{\textcolor{red}{\star}}$ denotes system crash due to memory.}\\  \cmidrule{2-8}
        & \pirname{} (\sect{sec:pirname})
        & 0.78 & 6.65 & 32.36 & 71.17 & 386.68 & 1542.30$^{\textcolor{red}{\star}}$\\ 
        \midrule \midrule
        \Block[c]{3-1}{Output Computation\tabularnote{includes message download, decryption/PIR queries, summation.} \\ (in sec)} 
        & \teename{} ~~(\sect{sec:teename})
        & 44.66 & 44.66 & 44.66 & 44.66 & 44.66 & 44.66 \\ \cmidrule{2-8}
        & $\PIRSUMA$ ~~(\sect{sec:pir-sumA})
        & 32.31 & 32.33 & 32.34 & 32.35 & 32.36 & 32.37\\\cmidrule{2-8}
        & $\PIRSUMB$ ~~(\sect{sec:pir-sumB})
        & 32.20 & 32.20 & 32.20 & 32.20 & 32.20 & 32.20\\
        \bottomrule
    \end{NiceTabular}
    \captionof{table}{Detailed computation costs per simulation ($\simsteps=14$, i.e., 14 days) in \ourname{}.}\label{tab:comp-micro-simstep}
    \vspace{10mm}
    \end{center}
\end{table*}

%% file: tex/7_acknowledgements.tex
\section*{Acknowledgements}
This project received funding from the European Research Council~(ERC) under the European Union's research and innovation programs Horizon~2020 (PSOTI/850990) and Horizon Europe (PRIVTOOLS/101124778). It was co-funded by the Deutsche Forschungsgemeinschaft~(DFG) within SFB~1119 CROSSING/236615297.

%% file: tex/8_app_related_work.tex
\iffullversion
\section{Related Primitives}
\label{app:related-primitives} 

In the following, we provide an overview about the (cryptographic) primitives and other techniques used in this work.

\subsection{Anonymous Communication}
\label{sec:anonycom}
To simulate the transmission of the modelled disease, \ourname{} requires anonymous messaging between participants. Mix-nets~\cite{chaum1981untraceable} and protocols based on the dining cryptographer (DC) problem \cite{chaum1988dining} were the first approaches to anonymous messaging. A fundamental technique underlying mix-nets is the execution of an oblivious shuffling algorithm that provides unlinkability between the messages before and after the shuffle. In a mix-net, so-called mix servers jointly perform the oblivious shuffling so that no single mix server is able to reconstruct the permutation performed on the input data. Past research established a wide variety of oblivious shuffle protocols based on garbled circuits~\cite{du2001study,wang2010bureaucratic,huang2012}, homomorphic encryption~\cite{huang2012}, distributed point functions~\cite{abraham2020}, switching networks~\cite{mohassel2013hide}, permutation matrices~\cite[§4.1]{laur2011round}, sorting algorithms~\cite[§4.2]{laur2011round}, and re-sharing~\cite[§4.3+4.4]{laur2011round}. Recently, the works of~\cite{Araki0OPRT21} and~\cite{eskandarian2021clarion} proposed efficient oblivious shuffling schemes using a small number of mix net servers.

\subsection{Trusted Execution Environment (TEE)} 
\teename{} (\sect{sec:teename}) assumes the availability of TEEs on the mobile devices of participants. TEEs are hardware-assisted environments providing secure storage and execution for sensitive data and applications isolated from the normal execution environment. Data stored in a TEE is secure even if the operating system is compromised, i.e., it offers confidentiality, integrity, and access control~\cite{ekberg2014,jauernig2020}.
Widely adopted TEEs are Intel SGX~\cite{sgx2014} and ARM TrustZone~\cite{trustzone2009} (often used on mobile platforms~\cite{ngabonziza2016}). Using TEEs for private computation has been extensively investigated, e.g.,~\cite{ohrimenko2016,Bayerl2020}. A process called remote attestation allows external parties to verify that its private data sent via a secure channel is received and processed inside the TEE using the intended code~\cite{chen2019,sgx}.

\subsection{Private Information Retrieval (PIR)}
The first computational single-server PIR (cPIR) scheme was introduced by Kushilevitz and Ostrovsky\ci{kushilevitz1997replication}. Recent cPIR schemes\ci{angel2018pir, gentry2019compressible} use homomorphic encryption (HE). However, single-server PIR suffers from significant computation overhead since compute intensive HE operations have to be computed on each of the database block for each PIR request. In contrast, multi-server PIR relies on a non-collusion assumption between multiple PIR servers and uses only XOR operations\ci{chor1995private,CGibbsB15,boyle2016fss,corrigan2020private,BonehBCGI21} making it significantly more efficient than cPIR.

\subsection{Cuckoo Hashing}
\label{sec:prelim_cuckoo}
In \pirname{} (\sect{sec:pirname}), messages of participants have to be stored in a database $D$. To do so, a hash function $H$ can be used to map an element $x$ into bins of the database: $D[H(x)]=x$. However, as we show in~\sect{sec:pirname}, \pirname{} requires that at most one element is stored in every database location which renders simple hashing impracticable~\cite{pinkas2015phasing}. Cuckoo hashing uses $h$ hash functions $H_1, \ldots, H_h$ to map elements into bins. It ensures that each bin contains exactly one element. If a collision occurs, i.e., if a new element is to be added into an already occupied bin, the old element is removed to make space for the new one. The evicted element, then, is placed into a new bin using another of the $h$ hash functions. If the insertion fails for a certain number of trials, an element is inserted into a special bin called stash which is the only one that is allowed to hold more than one element. Pinkas et al.~\cite{pinkas2015phasing} show that for $h=2$ hash functions and $n=2^{20}$ elements inserted to $2.4n$ bins, a stash size of 3 is sufficient to have a negligible error probability. 

\subsection{Garbled Cuckoo Table (GCT)}
\label{sec:prelim_gct}
As \pirname{} uses key-value pairs for the insertion into the database, a combination of garbled Bloom filters\ci{dong2013private} with cuckoo hashing\ci{pagh2004cuckoo, kirsch2010more}, called Garbled Cuckoo Table~\cite{pinkas2020psi}, is needed. Instead of storing $x$ elements in one bin as in an ordinary cuckoo table (cf.~\sect{sec:prelim_cuckoo}), in a GCT $h$ XOR shares of $x$ are stored at the $h$ locations determined by inputting $k$ into all $h$ hash functions. E.g., with $h=2$, if one of these two locations is already in use, the XOR share for the other (free) location is set to be the XOR of $x$ and the data stored in the used location. In \sect{app:agct}, we introduce a variant of GCT called \emph{arithmethic garbled cuckoo table} (AGCT) that uses arithmetic sharing over the ring $\mathbb{Z}_{2^\ell}$ instead of XOR sharing. For a database with $2.4n$ entries where $n$ is the number of elements inserted, Pinkas et al.~\cite{pinkas2020psi} show that the number of cycles is maximally $\log n$ with high probability. 

\subsection{Secure Multi-Party Computation (MPC)}
\label{sec:mpc}
MPC~\cite{yao86} allows a set of mutually distrusting parties to jointly compute an arbitrary function on their private inputs without leaking anything but the output. In the last years, MPC techniques in various security models have been introduced, extensively studied, and improved, e.g., in~\cite{damgard2013,lindell2015efficient,Demmler2015}. These advancements significantly enhance the efficiency of MPC making it more and more practical for real-world applications. Due to the practical efficiency it can provide, various works~\cite{BLAZE,FLASH,Trident,SWIFT,Araki0OPRT21,Tetrad} have recently concentrated on MPC for a small number of parties, especially in the three and four party honest majority setting tolerating one corruption. In \ourname{}, we employ MPC techniques across three servers to enable an anonymous communication channel (cf.~\sect{app:anon-comm-channel}) and to develop efficient $\PIRSUM$ protocols (cf.~\sect{sec:pirsum}).

\subsection{Anonymous Credentials}
\label{sect:anycred}
To protect against sybil attacks (cf.~\sect{sec:privacy-framework}), i.e., to hinder an adversary from creating multiple identities that can collect encounter information to detect correlations among unconscious encounters, we suggest to use anonymous credentials such that only registered participants can join \ourname. In this manner, the registration process can, for example, be linked to a passport. Such a registration system increases the cost to create (fake) identities.
Chaum\ci{chaum1985} introduced anonymous credentials where a client holds the credentials of several unlinkable pseudonyms. The client can then prove that it possesses the credentials of pseudonyms without the service provider being able to link different pseudonyms to the same identity. Additionally, anonymous credentials allow to certify specific properties like the age. Several instantiations for anonymous credentials have been proposed, e.g., Microsoft U-Prove\ci{uprove2013}.
\fi

%% file: tex/9_app_framework.tex
\iffullversion
\section{Building Blocks in \ourname}
\label{app:framework}
This section contains details about the building blocks used in the \ourname{} framework, such as shared-key setup, collision-resistant hash functions, anonymous communication channels, and Distributed Point Functions.

\subsection{Shared-Key Setup}
\label{app:shared-key} 
Let $F: \{0,1\}^{\kappa} \times  \{0,1\}^{\kappa} \rightarrow X$ be a secure pseudo-random function (PRF), with co-domain $X$ being $\Z{\ell}$ and $\mpcservers' = \mpcservers \cup \{\participant{i}\}$ for a participant $\participant{i} \in \participantset$. The following PRF keys are established among the parties in $\mpcservers'$ in \ourname:
\begin{itemize}
	\item[--] $k_{ij}$ among every $P_i, P_j \in \mpcservers'$ and $i \ne j$.
	\item[--] $k_{ijk}$ among every $P_i, P_j, P_k \in \mpcservers'$ and $i \ne j \ne k$.
	\item[--] $k_{\mpcservers'}$ among all the parties in $\mpcservers'$.
\end{itemize}
To sample a random value $r_{ij} \in \Z{\ell}$ non-interactively, each of $P_i$ and $P_j$ can invoke $F_{k_{ij}}(id_{ij})$. In this case, $id_{ij}$ is a counter that $P_i$ and $P_j$ maintain and update after each PRF invocation. The appropriate sampling keys are implied by the context and are, thus, omitted.

\subsection{Collision Resistant Hash Function}
\label{app:hash-function} 
A family of hash functions $\{ \Hash: \mathcal{K} \times \mathcal{L} \rightarrow \mathcal{Y}\}$ is said to be \emph{collision resistant} if, for all probabilistic polynomial-time adversaries $\Adv$, given the description of $\Hash_k$, where $k \in_R \mathcal{K}$, there exists a negligible function $negl()$ such that $\text{Pr}[(x, x') \leftarrow \Adv(k): (x \ne x') \wedge \Hash_k(x) = \Hash_k(x')] = negl(\kappa)$, where $x, x' \in_R \{0,1\}^m$ and $m = \text{poly}(\kappa)$.

\subsection{Anonymous Communication Channel}
\label{app:anon-comm-channel}
This section describes how to instantiate the $\funcAnon$ functionality used by \ourname{} for anonymous communication, as discussed in  \sect{sec:instan}. We start with the protocol for the case of \pirname{} and then show how to optimize it for the use in the \teename{} protocol. Recall from \sect{sec:pirname} that in \pirname, participants in $\participantset$ upload a set of messages from which a database $\db$ must be constructed at the end by $\mpcserver{1}$ and $\mpcserver{2}$. The anonymous communication is required to ensure that neither $\mpcserver{1}$ nor $\mpcserver{2}$ can link the source of the message even after receiving all messages in clear, which may not be in the same order. To tackle this problem, we use an approach based on oblivious shuffling inspired by~\cite{Araki0OPRT21,eskandarian2021clarion}, which is formalised next.

\emph{Problem Statement.} Consider the vector $\vec{m} = \{m_1, \ldots, m_{\threshold}\}$ of $\threshold$ messages with $m_j \in \Z{\ell}$ for $j \in [\threshold]$. We want servers $\mpcserver{1}$ and $\mpcserver{2}$ to obtain $\permutation(\vec{m})$, where $\permutation()$ denotes a random permutation that neither $\mpcserver{1}$ nor $\mpcserver{2}$ knows. Furthermore, an attacker with access to a portion of the network and, hence, the ability to monitor network data should not be able to gain any information about the permutation~$\permutation()$.

In \pirname, the vector $\vec{m}$ corresponds to the infection likelihood messages of the form $(a_{i,j},c^e_{i,j})$ that each participant $\participant{i} \in \participantset$ sends over the network (cf.~\sect{sec:pirname}). W.l.o.g., we let $\participant{i}$ have the complete $\vec{m}$ with them. The protocol makes use of the third server $\mpcserver{0}$ in our setting and proceeds as follows:

\begin{enumerate}
	\item[1.] $\participant{i}$ generates  an additive sharing of $\vec{m}$ among $\mpcserver{0}$ and $\mpcserver{1}$:
	\begin{acm_inneritem}
		\item[a)] $\participant{i}, \mpcserver{0}$ sample random $\vec{\ashr{m}{1}} \in \Z{\ell}^{\threshold}$.
		\item[b)] $\participant{i}$ computes and sends $\vec{\ashr{m}{2}} = \vec{m} - \vec{\ashr{m}{1}}$ to $\mpcserver{1}$.
	\end{acm_inneritem}
	\item[2.] $\mpcserver{0}$ and $\mpcserver{1}$ agree on a random permutation $\permutation_{01}$ and locally apply $\permutation_{01}$ to their shares. Let $\permutation_{01}(\vec{m}) = \permutation_{01}(\vec{\ashr{m}{1}}) + \permutation_{01}(\vec{\ashr{m}{2}})$.
	\item[3.] $\mpcserver{0}, \mpcserver{1}$ perform a \emph{re-sharing} of $\permutation_{01}(\vec{m})$, denoted by $\vec{m_{01}}$, by jointly sampling a random $\vec{r_{01}}  \in \Z{\ell}^{\threshold}$ and setting  $\vec{\ashr{m_{01}}{1}} = \permutation_{01}(\vec{\ashr{m}{1}}) + \vec{r_{01}}$ and $\vec{\ashr{m_{01}}{2}} = \permutation_{01}(\vec{\ashr{m}{2}}) - \vec{r_{01}}$.
	\item[4.] $\mpcserver{1}$ sends $\vec{\ashr{m_{01}}{2}}$ to $\mpcserver{2}$. Now, $(\vec{\ashr{m_{01}}{1}}, \vec{\ashr{m_{01}}{2}})$ forms an additive sharing of $\vec{m_{01}}$ among $\mpcserver{0}$ and $\mpcserver{2}$.
	\item[5.] $\mpcserver{0}$ and $\mpcserver{2}$ agree on a random permutation $\permutation_{02}$ and apply $\permutation_{02}$ to their shares. Let $\permutation_{02}(\vec{m_{01}}) = \permutation_{02}(\vec{\ashr{m_{01}}{1}}) + \permutation_{02}(\vec{\ashr{m_{01}}{2}})$.
	\item[6.] $\mpcserver{0}$ sends $\permutation_{02}(\vec{\ashr{m_{01}}{1}})$ to $\mpcserver{2}$, who reconstructs $\permutation_{02}(\vec{m_{01}})$.
	\item[7.] $\mpcserver{2}$ generates an \emph{additive-sharing} of $\permutation_{02}(\vec{m_{01}})$, denoted by $\vec{m_{02}}$, among $\mpcserver{1}$ and $\mpcserver{2}$, by jointly sampling $\vec{\ashr{m_{02}}{1}} \in \Z{\ell}^{\threshold}$ with $\mpcserver{1}$ and locally setting  $\vec{\ashr{m_{02}}{2}} = \permutation_{02}(\vec{m_{01}}) - \vec{\ashr{m_{02}}{1}}$.
	\item[8.] $\mpcserver{2}$ sends $\vec{\ashr{m_{02}}{2}}$ to $\mpcserver{1}$, who locally compute the output as $\vec{m_{02}} = \vec{\ashr{m_{02}}{1}} + \vec{\ashr{m_{02}}{2}}$.
\end{enumerate}

\paragraph{Anonymous Communication in \teename} As discussed in \sect{sec:teename}, the server $\mpcserver{2}$ is only required to have the complete set of messages in the clear but in an unknown random order. As a result, in the case of \teename, only the first permutation ($\permutation_{01}$ in Step 2) is sufficient and steps 5-8 are no longer required. Furthermore, in addition to the communication by $\mpcserver{1}$ in step 4, $\mpcserver{0}$ sends its share of $\vec{m_{01}}$ to $\mpcserver{2}$, who can then reconstruct $\vec{m_{01}} = \permutation_{01}(\vec{m})$.

\paragraph{Security Guarantees} As discussed in \sect{sec:threatmodel}, we assume that the MPC servers $\mpcserver{i}, i\in[2]$, that also instantiate the anonymous communication channel are semi-honest. We claim that the protocol described above will produce a random permutation of the vector $\vec{m}$ that neither $\mpcserver{1}$ nor $\mpcserver{2}$ is aware of. To see this, note that 
\begin{align*}
	\vec{m_{02}} = \permutation_{02}(\vec{m_{01}}) =  \permutation_{02}(\permutation_{01}(\vec{m}))
\end{align*}
and both $\mpcserver{1}$ and $\mpcserver{2}$ know only one of the two permutations $\permutation_{01}$ and $\permutation_{02}$, but not both. Furthermore, the re-sharing performed in step 3 and the generation of additive shares in step 6 above ensures that an attacker observing the traffic cannot relate messages sent and received.

As we also consider a client-malicious security model~\cite{lehmkuhl2021,chandran2021}, where some clients might deviate from the protocol to gain additional information, we also have to take into consideration how the clients could manipulate the communication to break anonymity. For \teename{}, this is trivial: The TEE ensures that clients' messages are correctly generated and uploaded. For \pirname{}, a malicious client could manipulate how many messages it uploads. However, messages with addresses that are already used will be dropped, i.e., effectively removing the malicious client from the system. A receiver will never fetch messages with unknown, random addresses. Furthermore, the servers use secure communication channels and even send freshly re-shared shares. We also consider a global attacker being able to monitor the full network traffic to be unrealistic. Hence, considering the discussed aspects/assumptions, classical attacks on anonymous communication such as flooding~\cite{Danezis2003} are not relevant for our model.

\subsection{Distributed Point Functions (DPF)}
\label{app:dpf}
Consider a point function $P_{\alpha,\beta}: \Z{\ell} \rightarrow \Z{\ell'}$ such that for all $\alpha \in \Z{\ell}$ and $\beta \in \Z{\ell'}$, $P_{\alpha, \beta}(\alpha) = \beta$ and $P_{\alpha,\beta}(\alpha') = 0$ for all $\alpha' \ne \alpha$. That is, when evaluated at any input other than $\alpha$, the point function $P_{\alpha, \beta}$ returns $0$ and when evaluated at $\alpha$ it returns $\beta$.

An $(s, t)$-distributed point function~(DPF)~\cite{GilboaI14,CGibbsB15} distributes a point function $P_{\alpha, \beta}$ among $s$ servers in such a way that no coalition of at most $t$ servers learns anything about $\alpha$ or $\beta$ given their $t$ shares of the function. We use $(2,1)$-DPFs in \ourname{} to optimize the communication of PIR-based protocols, as discussed in \sect{sec:reduce-comm-pir}. Formally, a $(2, 1)$-DPF comprises of the following two functionalities:

\medskip
\begin{itemize}
	\item[--] \Gen($\alpha, \beta$) $\rightarrow (\pirkey_1, \pirkey_2)$. Output two DPF keys $\pirkey_1$ and $\pirkey_2$, given $\alpha \in \Z{\ell}$ and $\beta \in \Z{\ell'}$. 
	\item[--] \Eval($\pirkey, \alpha') \rightarrow \beta'$. Return $\beta' \in \Z{\ell'}$, given key $\pirkey$ generated using \Gen, and an index $\alpha' \in \Z{\ell}$. 
\end{itemize}

\medskip
A $(2,1)$-DPF is said to be \emph{correct} if for all $\alpha, x \in \Z{\ell}$, $\beta \in \Z{\ell'}$, and $(\pirkey_1, \pirkey_2) \leftarrow$ \Gen$(\alpha, \beta)$, it holds that
\[
	\Eval(\pirkey_1, x) + \Eval(\pirkey_2, x) = (x = \alpha)~?~\beta : 0.
\]
A $(2,1)$-DPF is said to be \emph{private} if neither of the keys $\pirkey_1$ and $\pirkey_2$ leaks any information about $\alpha$ or $\beta$. That is, there exists a polynomial time algorithm that can generate a computationally indistinguishable view of an adversary $\Adv$ holding DPF key $\pirkey_u$ for $u \in \{1,2\}$, when given the key $\pirkey_u$.

As mentioned in \cite{CGibbsB15,boyle2016fss}, a malicious participant could manipulate the \Gen{} algorithm to generate incorrect DPF keys that do not correspond to any point function. While~\cite{CGibbsB15} used an external non-colluding auditor to circumvent this issue in the two server setting, \cite{boyle2016fss} formalised this issue and proposed an enhanced version of DPF called Verifiable DPFs.
In addition to the standard DPF, a verifiable DPF has an additional function called \Ver{} that can be used to ensure the correctness of the DPF keys. In contrast to \Eval, \Ver{} in a $(2,1)$-verifiable DPF is an interactive protocol between the two servers, with the algorithm returning a single bit indicating whether the input DPF keys $\pirkey_1$ and $\pirkey_2$ are valid. 

A verifiable DPF is said to be \emph{correct} if for all $\alpha \in \Z{\ell}$, $\beta \in \Z{\ell'}$, keys $(\pirkey_1, \pirkey_2) \leftarrow$ \Gen$(\alpha, \beta)$,  the verify protocol \Ver{} outputs $1$ with probability $1$. \Ver{} should ensure that no additional information about $\alpha$ or $\beta$ is disclosed to the party in possession of one of the DPF keys. Furthermore, the probability that \Ver{} outputs $1$ to at least one of the two servers for a given invalid key pair $(\pirkey'_1,\pirkey'_2)$ is negligible in the security parameter $\kappa$. 

Recent results in the area of (verifiable) DPFs~\cite{de2022lightweight,Kolobov22} might be an interesting direction for future work to further enhance the efficiency of our \pirname{} construction.

\paragraph{Communication Complexity}
Using the protocol of Boyle et. al.~\cite{boyle2016fss}, a $(2,1)$-DPF protocol for a point function with domain size $N$  has key size $(\lambda + 2) \cdot \log (N/\lambda) + 2 \cdot \lambda$ bits, where $\lambda = 128$ for an AES based implementation. The additional cost in the case of verifiable DPF is for executing the \Ver{} function, which has a constant number of elements in~\cite{boyle2016fss}. Furthermore, as stated in~\cite{boyle2016fss}, the presence of additional non-colluding servers can improve the efficiency of \Ver, and we use $\mpcserver{0}$ in the case of $\PIRSUMA$, as discussed in \sect{sec:pir-sumA}.  We refer to~\cite{boyle2016fss} for more details regarding the scheme.
\fi

%% file: tex/10_app_pir_sum_details_main.tex
\iffullversion
\section{PIR-SUM Protocol Details}
\label{app:PIR-SUM}
This section provides additional details of our $\PIRSUM$ protocols introduced in~\sect{sec:naive-pir}. We begin by recalling the security guarantees of a 2-server PIR for our setting~\cite{chor1995private,GHPS22}.
Informally in a two-server PIR protocol, where the database $\db$ is held by two non-colluding servers $\mpcserver{1}$ and $\mpcserver{2}$, a single server $\mpcserver{u} \in \{\mpcserver{1}, \mpcserver{2}\}$ should not learn any information about the client's query. The security requirement is formally captured in Definition~\ref{def:sec-pir-main}.

\begin{definition}
\label{def:sec-pir-main}
(Security of 2-server PIR) A PIR scheme with two non-colluding servers is called secure if each of the servers does not learn any information about the query indices.

Let $view(\mpcserver{u}, \querySet)$ denote the view of server $\mpcserver{u} \in \{\mpcserver{1}, \mpcserver{2}\}$ with respect to a list of queries, denoted by $\querySet$. We require that for any database $\db$, and for any two $\threshold$-length list of queries~$\querySet = (\query_1, \dots, \query_\threshold)$ and~$\querySet' = (\query_1^\prime, \dots, \query_\threshold^\prime)$, no algorithm whose run time is polynomial in $\threshold$ and in computational parameter~$\kappa$ can distinguish the view of the servers $\mpcserver{1}$ and $\mpcserver{2}$, between the case of participant $\participant{i}$ using the queries in $\querySet$ ($\{view(\mpcserver{u}, \querySet)\}_{u \in \{1,2\}}$), and the case of it using $\querySet'$ ($\{view(\mpcserver{u}, \querySet')\}_{u \in \{1,2\}}$).
\end{definition}

\subsection{Linear Summation PIR for $\funcPirS$ with optimized Communication.}
\label{app:linear-summation} 
This section describes the 2-server linear summation PIR protocol in~\cite{chor1995private}, as well as how to optimize communication using DPF techniques discussed in~Appendix~\ref{app:dpf}. 
To retrieve the $\query$-th block from database~$\db$ of size $\dbsize$, the linear summation PIR proceeds as follows:
\begin{itemize}
	\item Participant $\participant{i}$ prepares an $\dbsize$-bit string $\bitvector{}_{\query} = \{\bitV{1}{\query}, \ldots, \bitV{\dbsize}{\query}\}$ with $\bitV{j}{\query} = 1$ for $j = \query$ and $\bitV{j}{\query} = 0$ and $j \ne \query$, for $j \in [\dbsize]$. 
	\item $\participant{i}$ generates a Boolean sharing of $\bitvector{}_{\query}$ among $\mpcserver{1}$ and $\mpcserver{2}$, i.e., $\participant{i}, \mpcserver{1}$ non-interactively sample the random $[\bitvector{}_{\query}]_1 \in \{0,1\}^{\dbsize}$ and $\participant{i}$ sends $[\bitvector{}_{\query}]_2 = \bitvector{}_{\query} \xor [\bitvector{}_{\query}]_1$ to $\mpcserver{2}$.
	\item $\mpcserver{u}$, for $u \in \{1,2\}$, sends $[y]_u = \Xor\limits_{j = 1}^{\dbsize} [\bitV{j}{\query}]_u \dbV{j}$ to $\participant{i}$.
	\item $\participant{i}$ locally computes $\dbV{\query} = [y]_1 \xor [y]_2$.
\end{itemize}
The linear summation PIR described above requires communication of $\dbsize + 2\ell$ bits, where $\ell$ denotes the size of each data block in $\db$.

\subsubsection{Optimizing Communication using DPFs.}
\label{main:fss-pir-scheme}
Several works in the literature~\cite{GilboaI14,CGibbsB15,boyle2016fss,GHPS22} have used DPFs~(cf.~Appendix~\ref{app:dpf}) as a primitive to improve the communication in multi-server PIR. The idea is to use a DPF function to allow the servers $\mpcserver{1}$ and $\mpcserver{2}$ to obtain the XOR shares of an $\dbsize$-bit string $\bitvector{}$ that has a zero in all positions except the one representing the query $\query$. Because DPF keys are much smaller in size than the actual database size, this method aids in the elimination of $\dbsize$-bit communication from $\participant{i}$ to the servers, as in the aforementioned linear summation PIR. 

\noindent To query the $\query$-th block from a database $\db$ of size $\dbsize$, 
\begin{itemize}
	\item[--] Participant $\participant{i}$ executes the key generation algorithm with input $\query$ to obtain two DPF keys, i.e., $(\pirkey_1, \pirkey_2) \leftarrow$ \Gen($\query, 1$).
	\item[--] $\participant{i}$ sends $\pirkey_u$ to $\mpcserver{u}$, for $u \in \{1, 2\}$.
	\item[--] $\mpcserver{u}$, for $u \in \{1, 2\}$, performs a DPF evaluation at each of the positions $j \in [\dbsize]$ using key $\pirkey_u$ and obtains the XOR share corresponding to bit vector $\bitvector{}_{\query}$. 
	\begin{acm_inneritem}
		\item $\mpcserver{u}$ expands the DPF keys as $[\bitV{j}{\query}]_u \leftarrow$ \Eval($\pirkey_u, j$) for $j \in [\dbsize]$.
	\end{acm_inneritem}
	\item[--] $\mpcserver{u}$, for $u \in \{1,2\}$, sends $[y]_u = \Xor\limits_{j = 1}^{\dbsize} [\bitV{j}{\query}]_u \dbV{j}$ to $\participant{i}$.
	\item[--] $\participant{i}$ locally computes $\dbV{\query} = [y]_1 \xor [y]_2$.
\end{itemize}

\smallskip
For the case of semi-honest participants, we use the DPF protocol of~\cite{boyle2016fss}  and the key size is $O(\lambda \cdot \log(\dbsize/\lambda))$ bits, where $\lambda = 128$ is related to AES implementation in~\cite{boyle2016fss}. 

To prevent a malicious participant from sending incorrect or malformed keys to the servers, we use the verifiable DPF construction proposed in~\cite{boyle2016fss} for the case of malicious participants. This results only in a constant communication overhead over the semi-honest case. Furthermore, as noted in~\cite{boyle2016fss}, we use the additional server $\mpcserver{0}$ for a better instantiation of the verifiable DPF, removing the need for interaction with the participant $\participant{i}$ for verification. We provide more information in Appendix~\ref{app:dpf} and  refer the reader to~\cite{boyle2016fss} for all details.
\fi